\documentclass[11pt,reqno]{amsart}
\usepackage[body={14.5cm,20.4cm}]{geometry}
\usepackage{amssymb}
\usepackage{hyperref}
\usepackage{enumerate}
\usepackage{graphicx}
\usepackage{curves}
\usepackage{pstricks,pstricks-add}
\definecolor{gray50}{gray}{0.50}
\definecolor{gray75}{gray}{0.75}
\definecolor{gray85}{gray}{0.85}
\definecolor{gray90}{gray}{0.90}
\newcommand{\arxiv}[1]{\href{http://arxiv.org/abs/#1}{arXiv:#1}}
\newcommand*{\mailto}[1]{\href{mailto:#1}{\nolinkurl{#1}}}
\theoremstyle{plain}
\newtheorem{theorem}{Theorem}[section]
\newtheorem{lemma}[theorem]{Lemma}
\newtheorem{corollary}[theorem]{Corollary}
\theoremstyle{remark}

\newtheorem*{remarks*}{Remarks}

\providecommand{\D}{\mathbb}
\providecommand{\abs}[1]{\lvert#1\rvert}
\providecommand{\accol}[1]{\lbrace#1\rbrace}
\providecommand{\norm}[1]{\lVert#1\rVert}
\newcommand{\dd}{\mathrm{d}}
\newcommand{\ee}{\mathrm{e}}
\newcommand{\ii}{\mathrm{i}}
\renewcommand{\Im}{\operatorname{Im}}
\renewcommand{\Re}{\operatorname{Re}}
\newcommand{\id}{\mathbb{I}}
\DeclareMathOperator{\const}{Const}
\DeclareMathOperator{\ord}{O}
\DeclareMathOperator{\Res}{Res}
\DeclareMathOperator{\sol}{\textup{\textsf{sol}}}
\DeclareMathOperator{\wronsk}{\textsf{W}}
\numberwithin{equation}{section}
\begin{document}
\title[Long-Time Asymptotics for the CH Equation]{Long-Time Asymptotics for the Camassa--Holm Equation}
\author[A. Boutet de Monvel]{Anne Boutet de Monvel}
\address{Institut de Math\'ematiques de Jussieu, Universit\'e Paris Diderot Paris 7\\
175 rue du Chevaleret\\ 75013 Paris \\ France}
\email{\mailto{aboutet@math.jussieu.fr}}
\urladdr{\url{http://www.math.jussieu.fr/~aboutet/}}
\author[A. Kostenko]{Aleksey Kostenko}
\address{Institute of Applied Mathematics and Mechanics\\ NAS of Ukraine\\
R. Luxemburg str. 74\\ Donetsk 83114\\ Ukraine}
\email{\mailto{duzer80@gmail.com}}
\author[D. Shepelsky]{Dmitry Shepelsky}
\address{B.Verkin Institute for Low Temperature Physics\\
47 Lenin Avenue\\61103 Kharkiv\\Ukraine}
\email{\mailto{shepelsky@yahoo.com}}
\author[G. Teschl]{Gerald Teschl}
\address{Faculty of Mathematics\\
Nordbergstrasse 15\\ 1090 Wien\\ Austria\\ and International Erwin Schr\"odinger
Institute for Mathematical Physics\\ Boltzmanngasse 9\\ 1090 Wien\\ Austria}
\email{\mailto{Gerald.Teschl@univie.ac.at}}
\urladdr{\url{http://www.mat.univie.ac.at/~gerald/}}
\thanks{Research supported (G.T.) by the Austrian Science Fund (FWF) under Grant No.\ Y330.}
\thanks{SIAM J. Math. Anal. {\bf 41:4}, 1559--1588 (2009)}
\date{\today}
\begin{abstract}
We apply the method of nonlinear steepest descent to compute the long-time
asymptotics of the Camassa--Holm equation for decaying initial data, completing previous results by A.~Boutet de Monvel and D.~Shepelsky.
\end{abstract}
\maketitle
\section{Introduction}            \label{sec:intro}

In this paper we want to study the long-time asymptotics of the Camassa--Holm (CH) equation, also known as the dispersive shallow water equation,
\begin{equation}\label{ch.eq}
u_{t}+2\varkappa u_{x}-u_{txx}+3uu_x=2u_x u_{xx}+u u_{xxx},\qquad t>0,\qquad x\in\D{R},
\end{equation}
where $u\equiv u(x,t)$ is the fluid velocity in the $x$ direction, $\varkappa>0$ is a constant related to the critical shallow water
wave speed, and subscripts denote partial derivatives. It was first introduced by R.~Camassa and D.~Holm in \cite{ch}
and R.~Camassa et al.\ \cite{chh} as a model for shallow water waves, but already appeared earlier in a list by B.~Fuchssteiner and A.~Fokas \cite{ff}. More on the hydrodynamical relevance of this model can be found in the recent
articles by R.~Johnson \cite{jo} and A.~Constantin and D.~Lannes \cite{cla}. With
\begin{equation}
w:=u-u_{xx} +\varkappa,
\end{equation}
called the ``momentum'', equation \eqref{ch.eq} can be expressed as the condition of compatibility between
\begin{equation}\label{sp.ch}
\frac{1}{w} \left(-f''+\frac{1}{4}f\right)=\lambda f,
\end{equation}
and
\begin{equation}\label{LA-pair}
\partial_tf=-\left(\frac{1}{2\lambda}+u\right)f'+\frac{1}{2}u'f,
\end{equation}
that is,
\[
\partial_t\partial_{xx}f=\partial_{xx}\partial_tf
\]
is the same as saying that \eqref{ch.eq} holds. Equation \eqref{sp.ch} is the spectral problem associated to \eqref{ch.eq}.
In particular, the CH equation is completely integrable and can be solved via the inverse scattering method. Correspondingly,
we consider real-valued classical solutions $u(x,t)$ of the CH equation \eqref{ch.eq}, which decay rapidly,
that is,
\begin{equation}\label{decay}
\max_{0<t\leq T} \int_{\D{R}} (1+|x|)^{l+1} \big(|w(x,t)-\varkappa| + |w_x(x,t)| +|w_{xx}(x,t)|\big)\dd x < \infty
\end{equation}
for all $T>0$ and some integer $l\geq 1$. Moreover, we will assume
\begin{equation}
w(x,0)>0
\end{equation}
throughout this paper. Then $u$ exists for all times $t>0$ with $w(x,t)>0$ (for existence of solutions we refer to A.~Constantin and J.~Escher \cite{ce} and the discussion in A.~Constantin and J.~Lenells \cite{cl}, see also \cite{co}).

The aim of this paper is to establish the long-time asymptotics of such solutions using the nonlinear steepest descent method from P.~Deift and X.~Zhou \cite{dz} which was inspired by earlier work of S.~Manakov \cite{ma} and A.~Its \cite{its1}. More on this method and its
history can be found in the survey by P.~Deift, A.~Its, and X.~Zhou \cite{diz}.

The starting point for this method is the representation of a solution of the nonlinear equation under consideration in terms of a solution of an associated Riemann-Hilbert problem.

Recently, A.~Boutet de Monvel and D.~Shepelsky have used the inverse scattering approach to the CH equation, based on the construction and analysis of an associated matrix Riemann-Hilbert problem \cite{bms,bms2}. The analysis, via the nonlinear steepest descent method, of a vector oscillatory RH problem derived from the original matrix-valued problem, allowed \cite{bms3} distinguishing four main regions in the $(x,t)$-half-plane, where the leading asymptotic terms were qualitatively different: a solitonic sector, two sectors of (slowly decaying) modulated oscillations, and a sector of rapid decay.

Here we want to simplify the original approach, deriving the vector Riem\-ann-Hilbert problem directly from scattering theory for the underlying Sturm--Liouville operator. We notice that the matrix and vector RH problems, being closely related, have specific features concerning, particularly, the uniqueness issue. For the matrix problem, we refer to [5], whereas the uniqueness for the vector problem is addressed in detail in the present paper, see Section \ref{sec:istrhp} below. 

At the same time we want to provide complete proofs including the effects of solitons in the sectors of decaying oscillations (formulas given in \cite{bms3} for the oscillatory regions are true in the solitonless case only) and the error estimates in terms of decay of the initial condition.
Moreover, we will make use of some
simplifications to the nonlinear steepest descent method recently given in H.~Kr\"uger and G.~Teschl \cite{kt} respectively K.~Grunert and G.~Teschl \cite{gt}.

The asymptotics in the transition regions, near the lines $x/\varkappa t=2$ and $x/\varkappa t=-1/4$, involves Painlev\'e transcendents. A detailed analysis is presented in \cite{bmis}.
\section{Main result}             \label{sec:main.result}

In order to state our main result we first need to recall a few things.

\subsection*{Scattering data}
Associated with $w(x,t)$ is a self-adjoint Sturm--Liouville operator
\begin{equation}\label{defslop}
\begin{split}
&H(t)=\frac{1}{w(x,t)}\left(-\frac{\dd^2}{\dd x^2}+\frac{1}{4}\right),\\
&D\bigl(H(t)\bigr)=H^2(\D{R})\subset L^2(\D{R},w\,\dd x).
\end{split}
\end{equation}
Here $L^2(\D{R},w\,\dd x)$ denotes the Hilbert space of square integrable (complex-valued) functions over $\D{R}$ and $H^2(\D{R})$ the corresponding Sobolev space. 

By our assumption \eqref{decay} the spectrum of $H(t)$ (independent of $t$) consists of an absolutely continuous part $[\frac{1}{4\varkappa},\infty)$ plus a finite set of eigenvalues $\lambda_j\in(0,\frac{1}{4\varkappa})$, $1\leq j\leq N$, possibly empty (``solitonless case''), see \cite[Theorem 2.1]{co}. For $j=1,\dots,N$ we denote
\begin{align*}
&\lambda_j = \frac{1}{\varkappa}\Bigl(\frac{1}{4}-\kappa_j^2\Bigr),\\
&0<\kappa_1 < \kappa_2 < \cdots < \kappa_N <\frac{1}{2}\,.
\end{align*} 
Moreover, associated with the continuous spectrum is a (right)
reflection coefficient $R_+(k,t)$, and associated with every eigenvalue $\lambda_j$
is a (right) norming constant $\gamma_{+,j}(t)$ (see the next section for details).

\subsection*{One-soliton solution}
Given $\kappa\in(0,\frac{1}{2})$ and $\gamma>0$ we have the corresponding one-soliton solution $u_{\sol}^{(\kappa,\gamma)}\equiv u_{\sol}(x,t)$ given by
\begin{equation}
u_{\sol}(x-\varkappa c t) = 
\frac{32 \varkappa \kappa^2}{(1-4\kappa^2)^2}\;\frac{\alpha(y(x-\varkappa ct))}%
{(1+\alpha(y(x-\varkappa c t)))^2+\frac{16\kappa^2}{1-4\kappa^2}\ \alpha(y(x-\varkappa c t))},
\end{equation}
with
\begin{subequations}
\begin{align}
\alpha(y) 
&= \frac{\gamma^2}{2\kappa}\ee^{-2 \kappa y},\\
c 
&= \frac{1}{2(\tfrac{1}{4}-\kappa^2)} \in (2,\infty),
\end{align}
\end{subequations}
where $y(x)$ is given implicitly by
\begin{equation}
x = y + \log\frac{1+\alpha(y) \frac{1+2\kappa}{1-2\kappa}}{1+\alpha(y) \frac{1-2\kappa}{1+2\kappa}}\,.
\end{equation}
The momentum $w_{\sol}^{(\kappa,\gamma)}\equiv w_{\sol}(x,t)$ of this solution is given by
\begin{equation}
w_{\sol}(x-\varkappa ct) 
= \varkappa\biggl(1+\frac{16\kappa^2}{1-4\kappa^2}\frac{\alpha(y(x-\varkappa ct))}{(1+\alpha(y(x-\varkappa ct)))^2}\biggr)^2.
\end{equation}
Note that the one-soliton solution has the form of a single peak
which is symmetric
with respect to its center 
\[
x_0 =
\frac{1}{2\kappa}\log\frac{\gamma^2}{2\kappa}+
\log\frac{1+2\kappa}{1-2\kappa}\,.
\]
The maximum at $x_0$ is given by
$u_{\sol,\max}=\varkappa\frac{8\kappa^2}{1-4\kappa^2}$ and taller waves  are
narrower and travel
faster (for fixed $\varkappa$). See Fig.~\ref{fig:one-soliton}, where $u_{\sol}(x-x_0)$ for fixed $\varkappa$ and different values of $\kappa$ is displayed.

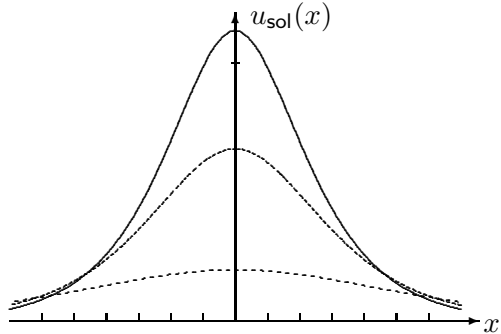
\begin{figure}[ht]
\begin{picture}(6.4,4.1)

\put(0,-0.156){\vector(1,0){6.24}}
\put(3.,-0.156){\vector(0,1){4.1}}
\put(6.3,-0.3){$x$}
\put(3.2,3.8){$u_{\sol}(x)$}

\multiput(0.429,-0.156)(0.4285,0){13}{\line(0,1){0.1}}
\put(2.95,3.27847){\line(1,0){0.1}}

\curve(0,0, 0.3,0.08, 0.6,0.199, 0.9,0.375, 1.2,0.632, 1.5,0.997, 
1.8,1.497, 2.1,2.133, 2.4,2.844, 2.7,3.457, 3.,3.708, 3.3,3.457, 
3.6,2.844, 3.9,2.133, 4.2,1.497, 4.5,0.997, 4.8,0.632, 5.1,0.375, 
5.4,0.199, 5.7,0.08, 6.,0)
\curvedashes{0.02,0.05}
\curve(0,0.05, 0.3,0.132, 0.6,0.244, 0.9,0.395, 1.2,0.593, 1.5,0.844, 
1.8,1.147, 2.1,1.481, 2.4,1.802, 2.7,2.043, 3.,2.134, 3.3,2.043, 
3.6,1.802, 3.9,1.481, 4.2,1.147, 4.5,0.844, 4.8,0.593, 5.1,0.395, 
5.4,0.244, 5.7,0.132, 6.,0.05)
\curvedashes{0.05,0.05}
\curve(0,0.094, 0.3,0.139, 0.6,0.189, 0.9,0.242, 1.2,0.298, 
1.5,0.355, 1.8,0.408, 2.1,0.455, 2.4,0.492, 2.7,0.516, 3.,0.524, 
3.3,0.516, 3.6,0.492, 3.9,0.455, 4.2,0.408, 4.5,0.355, 4.8,0.298, 
5.1,0.242, 5.4,0.189, 5.7,0.139, 6.,0.094)
\end{picture}
\caption{One-solitons for fixed $\varkappa=1$ and $\kappa=.15,\,.25,\,.3$}
\label{fig:one-soliton}
\end{figure}

\noindent
Note also that if $\kappa=\kappa(\varkappa)$ is changing in such a way that the \emph{soliton velocity} 
\[
v=c\varkappa=\frac{2\varkappa}{1-4\kappa^2}
\]
is fixed,
then, as $\varkappa\to 0$, the form of the soliton approaches that of a \emph{peakon}, $v\ee^{-|x|}$, which is a nonsmooth, weak solution of the Camassa--Holm equation with $\varkappa=0$ (see~\cite{bms2}). See Fig.~\ref{fig:fixed-velocity}, where $u_{\sol}(x-x_0)$ for fixed velocity $v$ and different values of $(\varkappa,\,\kappa)$ is displayed. Note that $u_{\sol,\max}=v-2\varkappa$.

\begin{figure}[ht]
\begin{picture}(6.24,4)

\put(0,0){\vector(1,0){6.24}}
\put(3.,0){\vector(0,1){4}}
\put(6.3,-0.15){$x$}
\put(3.2,3.8){$u_{\sol}(x)$}

\multiput(0.429,0)(0.4285,0){13}{\line(0,1){0.1}}
\multiput(2.95,1.615)(0,1.615){2}{\line(1,0){0.1}}

\curve(0,0, 0.2,0.003, 0.4,0.008, 0.6,0.016, 0.8,0.027, 1.,0.046, 
1.2,0.075, 1.4,0.121, 1.6,0.192, 1.8,0.302, 2.,0.475, 2.2,0.744, 
2.4,1.16, 2.6,1.798, 2.8,2.743, 3.,3.708, 3.2,2.743, 3.4,1.798, 
3.6,1.16, 3.8,0.744, 4.,0.475, 4.2,0.302, 4.4,0.192, 4.6,0.121, 
4.8,0.075, 5.,0.046, 5.2,0.027, 5.4,0.016, 5.6,0.008, 5.8,0.003, 6.,0)
\curvedashes{0.02,0.05}
\curve(0,0.019, 0.24,0.032, 0.48,0.053, 0.72,0.085, 0.96,0.134, 
1.2,0.209, 1.44,0.323, 1.68,0.495, 1.92,0.75, 2.16,1.118, 2.4,1.619, 
2.64,2.221, 2.88,2.717, 3.12,2.717, 3.36,2.221, 3.6,1.619, 
3.84,1.118, 4.08,0.75, 4.32,0.495, 4.56,0.323, 4.8,0.209, 5.04,0.134, 
5.28,0.085, 5.52,0.053, 5.76,0.032, 6.,0.019)
\curvedashes{0.05,0.05}
\curve(0,0.103, 0.3,0.14, 0.6,0.187, 0.9,0.248, 1.2,0.324, 1.5,0.413, 
1.8,0.514, 2.1,0.618, 2.4,0.711, 2.7,0.778, 3.,0.802, 3.3,0.778, 
3.6,0.711, 3.9,0.618, 4.2,0.514, 4.5,0.413, 4.8,0.324, 5.1,0.248, 
5.4,0.187, 5.7,0.14, 6.,0.103)
\end{picture}
\caption{Fixed velocity $v=2.5$ and $(\varkappa,\kappa)=(1,.22),\,(.5,.39),\,(.1,\,.48)$}
\label{fig:fixed-velocity}
\end{figure}
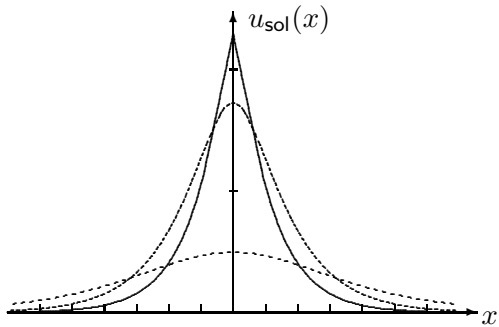

Now we are ready to state our main result.

\begin{theorem}[solution asymptotics]\label{thm:main}
Suppose $u(x,t)$ is a classical solution of the CH equation \eqref{ch.eq} satisfying \eqref{decay} for some integer $l\geq 1$. Let 
\[
R(k)=R_+(k,0)\text{ and }\kappa_j,\ \gamma_j=\gamma_{+,j}(0),\ j=1,\dots,N, 
\]
be the (right) scattering data
associated with $H(0)$ and the initial condition $w(x,0)$.

There are four sectors in the $(x,t)$ half-plane in which the long-time asymptotics of a solution of the CH equation satisfying \eqref{decay} are given by the following formulas:
\begin{enumerate}[\rm(i)]
\item
The ``soliton'' region $c:=\frac{x}{\varkappa t}>2+C$ for any small $C>0$. Let 
\[
c_j=\frac{1}{2(\frac{1}{4}-\kappa_j^2)},\quad j=1,\dots,N,
\]
and let $\varepsilon>0$ sufficiently small such that the intervals $[c_j-\varepsilon,c_j+\varepsilon]$ are disjoint and lie inside $(2,\infty)$.
\begin{enumerate}[\rm{(\theenumi}$_1$)]
\item
If $|\frac{x}{\kappa t}-c_j|<\varepsilon$ for some $j$, one has
\begin{equation}
u(x,t)= u_j(x-\xi_j-\varkappa c_jt)+\ord(t^{-l}),
\end{equation}
where $u_j=u_{\sol}^{(\kappa_j,\gamma_j)}$ is the one-soliton solution formed from $\kappa_j$ and
\begin{equation}     \label{eq:tilde.gamma.j}
\hat\gamma_j=\gamma_j\prod_{i=j+1}^N\frac{\kappa_i-\kappa_j}{\kappa_i+\kappa_j}\,,
\end{equation}
and having a phase shift
\begin{equation}     \label{eq:xi.j}
\xi_j=2\sum_{i=j+1}^N \log\frac{1+2\kappa_i}{1-2\kappa_i}\,.
\end{equation}
For $j=N$ the product is $1$ in \eqref{eq:tilde.gamma.j} and the sum is $0$ in \eqref{eq:xi.j}.
\item
If $\bigl\lvert\frac{x}{\kappa t}-c_j\bigr\rvert\geq\varepsilon$ for all $j$, one has
\begin{equation}
u(x,t) = \ord(t^{-l}).
\end{equation}
\end{enumerate}
\item
The ``first oscillatory'' region $0 \leq c:=\frac{x}{\varkappa t}< 2-C$ for any $C>0$. Here
\begin{align}   \label{first-oscill}
u(x,t) 
&=  -\sqrt{\frac{2\varkappa k_0(c) \nu_0(c)}{(\frac{1}{4}+k_0(c)^2)(\frac{3}{4}-k_0(c)^2) t}}
\sin\left(\frac{2\varkappa k_0(c)^3}{(\frac{1}{4}+k_0(c)^2)^2}t - \nu_0(c) \log(t) + \delta_0(c) \right)\notag\\
&\quad
+ \ord(t^{-\alpha})
\end{align}
for any $\frac{1}{2}<\alpha<1$ provided $l\geq 5$, where
\begin{align}
k_0(c) 
&=\frac{1}{2} \sqrt{-\frac{1+c-\sqrt{1+4c}}{c}},\\
\nu_0(c)
&=-\frac{1}{2\pi}\log(1-\left\vert R(k_0(c)) \right\vert^2),\\ \notag
\delta_0(c) 
&=\frac{\pi}{4}-\arg(R(k_0(c)))+\arg(\Gamma(\ii\nu_0(c)))\\ \notag
&\quad 
-\nu_0(c)\log\left(\frac{8\varkappa k_0(c)^2(3/4-k_0(c)^2)}{(1/4+k_0(c)^2)^3}\right) \\ \notag
&\quad
+4\sum_{j=1}^N\arctan\bigl(\frac{\kappa_j}{k_0(c)}\bigr) +
\frac{1}{\pi}\int_{\Sigma(c)}\log(\left\lvert\zeta-k_0(c)\right\rvert)\,\dd\log(1-\left\vert R(\zeta)\right\vert^2)\\
&\quad
+ 4\varkappa k_0(c) \sum_{j=1}^{N} \log\left(\frac{1+2\kappa_j}{1-2\kappa_j}\right)
+ \frac{4\varkappa k_0(c)}{\pi}\int_{\Sigma(c)}\frac{\log(|T(\zeta)|^2)}{1+4\zeta^2}\,\dd\zeta,
\end{align}
with $\Sigma(c)=(-k_0(c),k_0(c))$ and $\Gamma$ the Gamma function.
\begin{figure}[ht]\rm
\psset{unit=.9}
\begin{pspicture}(-7,-.2)(7,7)
\pspolygon[linestyle=none,fillstyle=solid,fillcolor=gray90](0,0)(0,6)(-.06,5.98)
\psline[linewidth=.5pt,linecolor=gray75]{->}(-6,0)(7.15,0)
\psline[linestyle=dashed](-6,0)(7,0)
\psline[linewidth=.5pt,linecolor=gray75]{->}(0,-.6)(0,6.2)
\psline[linestyle=dashed](0,0)(0,6.2)
\pspolygon[linestyle=none,fillstyle=solid,fillcolor=gray90](0,0)(6,2.8)(6,3.2)
\psline[linestyle=dashed](0,0)(6,3)
\psline[linewidth=.5pt,linecolor=gray85](0,0)(6,3.2)
\psline[linewidth=.5pt,linecolor=gray85](0,0)(6,2.8)
\pspolygon[linestyle=none,fillstyle=solid,fillcolor=gray90](0,0)(6,2.3)(6,2.5)
\psline[linestyle=dashed,linecolor=gray50](0,0)(6,2.4)
\psline[linewidth=.5pt,linecolor=gray85](0,0)(6,2.5)
\psline[linewidth=.5pt,linecolor=gray85](0,0)(6,2.3)
\pspolygon[linestyle=none,fillstyle=solid,fillcolor=gray90](0,0)(6,1.4)(6,1.6)
\psline[linestyle=dashed,linecolor=gray50](0,0)(6,1.5)
\psline[linewidth=.5pt,linecolor=gray85](0,0)(6,1.6)
\psline[linewidth=.5pt,linecolor=gray85](0,0)(6,1.4)
\pspolygon[linestyle=none,fillstyle=solid,fillcolor=gray90](0,0)(6,.6)(6,.8)
\psline[linestyle=dashed,linecolor=gray50](0,0)(6,.7)
\psline[linewidth=.5pt,linecolor=gray85](0,0)(6,.8)
\psline[linewidth=.5pt,linecolor=gray85](0,0)(6,.6)
\pspolygon[linestyle=none,fillstyle=solid,fillcolor=gray90](0,0)(-1.3,6)(-1.7,6)
\psline[linestyle=dashed](0,0)(-1.5,6)
\psline[linewidth=.5pt,linecolor=gray85](0,0)(-1.7,6)
\psline[linewidth=.5pt,linecolor=gray85](0,0)(-1.3,6)
\rput(0,0){$\bullet$}
\rput(.3,-.3){$0$}
\rput(6.75,.7){\footnotesize$c=c_N$}
\rput(6.7,1.5){\footnotesize$c=c_j$}
\rput(6.7,2.4){\footnotesize$c=c_1$}
\rput(6.6,3.15){$c=2$}
\rput(.65,5.4){$c=0$}
\rput(-2.2,6.4){$c=-\frac{1}{4}$}
\rput(4.5,1.45){(i)}
\rput(4.5,.85){\scriptsize solitons}
\rput(2.5,4){(ii)}
\rput(2.5,3.6){\scriptsize$1^{\text{st}}$ oscillatory}
\rput(2.5,3.2){\scriptsize region}
\rput(-.55,5){(iii)}
\rput(-.1,6.8){\scriptsize$2^{\text{nd}}$ oscillatory}
\rput(-.8,6.4){\scriptsize region}
\psline[linewidth=.5pt]{->}(-.7,6.1)(-.6,5.3)
\rput(-3.5,2.6){(iv)}
\rput(-3.5,2.2){\scriptsize fast decay}
\rput(6.8,-.2){$x$}
\rput(.2,6.1){$t$}
\end{pspicture}
\caption{The different regions of the $(x,t)$-half-plane, $c=\frac{x}{\varkappa t}$} 
\label{fig.regions}
\end{figure}
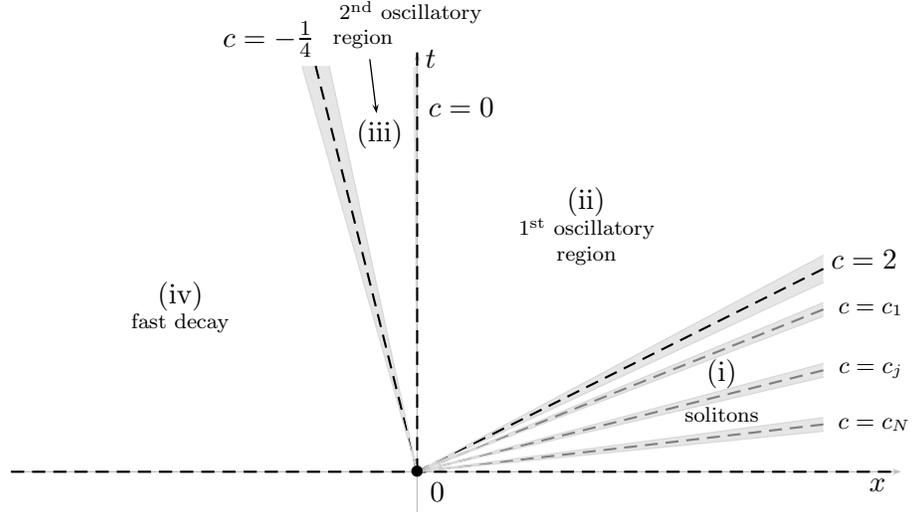
\item
The ``second oscillatory'' region $-\frac{1}{4}+C<c=\frac{x}{\varkappa t}< 0$ for any $C>0$, where
\begin{align}    \label{second-oscill}
u(x,t) 
&=-\sqrt{\frac{2\varkappa k_0(c) \nu_0(c)}{(\frac{1}{4}+k_0(c)^2)(\frac{3}{4}-k_0(c)^2) t}}
\sin\left(\frac{2\varkappa k_0(c)^3}{(\frac{1}{4}+k_0(c)^2)^2}t-\nu_0(c)\log(t) + \delta_0(c) \right)\notag\\
&\quad
-\sqrt{\frac{2\varkappa k_1(c) \nu_1(c)}{(\frac{1}{4}+k_1(c)^2)(k_1(c)^2-\frac{3}{4}) t}}
\sin\left(\frac{2\varkappa k_1(c)^3}{(\frac{1}{4}+k_1(c)^2)^2}t + \nu_1 \log(t) - \delta_1(c) \right) \notag\\
&\quad + \ord(t^{-\alpha})
\end{align}
for any $\frac{1}{2}<\alpha<1$ provided $l\geq 5$. Here, for $\ell=0,1$,
\begin{align}   \label{k-ell}
k_{\ell}(c) 
&= \frac{1}{2} \sqrt{-\frac{1+c-(-1)^{\ell}\sqrt{1+4c}}{c}}\,,\\
\nu_{\ell}(c)
&= -\frac{1}{2\pi}\log(1-\left\vert R(k_{\ell}(c)) \right\vert^2),\\ \notag
\delta_{\ell}(c) 
&= \frac{\pi}{4} - (-1)^{\ell}\arg(R(k_{\ell}(c)))+\arg(\Gamma(\ii\nu_{\ell}(c)))\\ \notag
&\quad 
-\nu_{\ell}(c) \log\frac{(-1)^{\ell}8\varkappa k_{\ell}(c)^2(3/4-k_{\ell}(c)^2)}{(1/4+k_{\ell}(c)^2)^3}
+ (-1)^{\ell}4\sum_{j=1}^N\arctan\frac{\kappa_j}{k_{\ell}(c)}\\
\notag
&\quad
+\frac{(-1)^{\ell}}{\pi}\int_{\Sigma(c)} \log(\left\vert \zeta-k_{\ell}(c) \right\vert)\dd\log(1-\left\vert R(\zeta) \right\vert^2)\\ \notag
&\quad
+(-1)^{\ell}2\nu_{\ell'}(c)\log\frac{k_1(c)-k_0(c)}{k_1(c)+k_0(c)}\\
&\quad
+ 4\varkappa k_{\ell}(c) \sum_{j=1}^N\log\frac{1+2\kappa_j}{1-2\kappa_j}
+ \frac{4\varkappa k_{\ell}(c) }{\pi}\int_{\Sigma(c)}\frac{\log(|T(\zeta)|^2)}{1+4\zeta^2}\,\dd\zeta,
\end{align}
with $\Sigma(c)=(-\infty,-k_1(c))\cup(-k_0(c),k_0(c))\cup(k_1(c),\infty)$, $0'=1$, and $1'=0$.
\item
The ``fast decay'' region $c=\frac{x}{\varkappa t}<-\frac{1}{4}-C$ for any $C>0$, where
\begin{equation}
u(x,t)=\ord(t^{-l}).
\end{equation}
\end{enumerate}
\end{theorem}

We can also give the asymptotics of the momentum $w(x,t)$.

\begin{theorem}[momentum asymptotics]\label{thm:main.momentum}
Under the same hypotheses and with the same notations the momentum of the solution behaves as follows:
\begin{enumerate}[\rm(i)]
\item
In the soliton region $c=\frac{x}{\varkappa t}>2+C$.
\begin{enumerate}[\rm{(\theenumi}$_1$)]
\item
If $|\frac{x}{\varkappa t}-c_j|<\varepsilon$ for some $j$, one has
\begin{equation}
w(x,t)=w_j(x-\xi_j-\varkappa c_jt)+\ord(t^{-l})
\end{equation}
where $w_j$ is the momentum of the one-soliton solution formed from the same parameters $\accol{\kappa_j,\gamma_j}$ and with the same phase shift $\xi_j$ as above.
\item
If $|\frac{x}{\varkappa t}-c_j|\geq\varepsilon$ for all $j$, one has
\begin{equation}
w(x,t)=\varkappa+\ord(t^{-l}).
\end{equation}
\end{enumerate}
\item
In the first oscillatory region $0\leq c=\frac{x}{\varkappa t}<2-C$ for any $C>0$, one has
\begin{align}\notag
w(x,t) 
&= \varkappa - 4\sqrt{\frac{2\varkappa k_0(c)(\frac{1}{4}+k_0(c)^2) \nu_0(c)}{(\frac{3}{4}-k_0(c)^2) t}}
\sin\left(\frac{2\varkappa k_0(c)^3}{(\frac{1}{4}+k_0(c)^2)^2}t - \nu_0(c) \log(t) + \delta_0(c) \right)\\
&\quad + \ord(t^{-\alpha}),
\end{align}
for any $\frac{1}{2}<\alpha<1$ provided $l\geq 5$.
\item
In the second oscillatory region $-\frac{1}{4}+C<c=\frac{x}{\varkappa t}<0$ for any $C>0$, one has
\begin{align}\notag
w(x,t) 
&=\varkappa 
-4\sqrt{\frac{2\varkappa k_0(c)(\frac{1}{4}+k_0(c)^2) \nu_0(c)}{(\frac{3}{4}-k_0(c)^2) t}}
\sin\left(\frac{2\varkappa k_0(c)^3}{(\frac{1}{4}+k_0(c)^2)^2}t - \nu_0(c) \log(t) + \delta_0(c)\right)\\ 
&\quad
- 4\sqrt{\frac{2\varkappa k_1(c)(\frac{1}{4}+k_1(c)^2) \nu_1(c)}{(k_1(c)^2-\frac{3}{4})t}}
\sin\left(\frac{2\varkappa k_1(c)^3}{(\frac{1}{4}+k_1(c)^2)^2}t + \nu_1(c) \log(t) - \delta_1(c)\right)\notag\\
&\quad
+\ord(t^{-\alpha}),
\end{align}
for any $\frac{1}{2}<\alpha<1$ provided $l\geq 5$.
\item
In the fast decay region $\frac{x}{\varkappa t}<-\frac{1}{4}-C$ for any $C>0$, we have
\begin{equation}
w(x,t)=\varkappa+\ord(t^{-l}).
\end{equation}
\end{enumerate}
\end{theorem}

In particular we recover the fact that a pure soliton solution (i.e., $R(k)\equiv 0$) asymptotically splits into single solitons with associated phase shifts. This was shown only recently by R.S.~Johnson \cite{jo2} (for two solitons) and in the general case by Y.~Matsuno \cite{mat}. For further results on solitons of the CH equation and their stability we refer to A.~Constantin and W.~Strauss \cite{cs}, L.-C.~Li \cite{lc}, and K.~El Dika and L.~Molinet \cite{em}.

Notice that the oscillatory regions (ii) and (iii) match at $x=0$. Indeed, as $x\to 0$ with $x<0$, $k_1\to\infty$ in \eqref{k-ell} and thus the amplitude of the second term in \eqref{second-oscill} vanishes, while the parameters of the first term in \eqref{second-oscill} match those in \eqref{first-oscill}.
As for the transition between the other regions, we have already noticed
in Section~\ref{sec:intro} that there exist transition zones, where the asymptotics are described in terms of Painlev\'e transcendents. More precisely (details are given in \cite{bmis}), these zones are:
\begin{enumerate}[{(tr}1)]
\item
$\abs{x/\varkappa t-2}\,t^{2/3}<\const$,
\item
$\abs{x/\varkappa t + 1/4}\,t^{2/3}<\const$.
\end{enumerate}

It should also be emphasized that unlike the (modified) Korteweg--de Vries (KdV) equation (originally considered in \cite{dz}), the asymptotic form is given implicitly, however, to leading order this fact only manifests itself in additional phase shift. More precisely, the term $\xi_j$ in the soliton region (as already pointed out in \cite{mat}) and the last two terms in $\delta_j(c)$, have no analog in the (modified) KdV equation (cf.~\cite{dz}, respectively \cite{gt}). For results on CH on the half-line we refer to A.~Boutet de Monvel and D.~Shepelsky \cite{bms0,bms4}.

Finally, note that if $u(x,t)$ solves the CH equation, then so does $u(-x,-t)$. Therefore it suffices to investigate the case $t\to+\infty$.
\section{The Inverse scattering transform and the Riemann--Hilbert problem}
\label{sec:istrhp}

In this section we derive a vector Riemann--Hilbert problem directly from the
scattering theory for the differential operator \eqref{sp.ch}. 
We begin by recalling some
required results from scattering theory, respectively the inverse
scattering transform for the CH equation from \cite{co,cgi} (see also \cite{mar}).

Recall also that by virtue of the unitary Liouville transform
\begin{equation}\label{liovtrf}
\begin{split}
&f(x) \mapsto \tilde{f}(y) = w(x)^{1/4}f(x),\\
&y=x-\int_x^{+\infty}\biggl(\sqrt{\frac{w(r)}{\varkappa}}-1\biggr)\dd r,
\end{split}
\end{equation}
the Sturm--Liouville operator $H(t)$ introduced in \eqref{defslop} can be mapped to a self-adjoint Schr{\"o}dinger operator
\begin{equation}\label{schroeop}
\begin{split}
&\tilde{H}(t) = \frac{1}{\varkappa}\biggl(-\frac{\dd^2}{\dd y^2}+q(\,\cdot\,,t) + \frac{1}{4}\biggr), \\
&D\bigl(\tilde{H}(t)\bigr)=H^2(\D{R})\subset L^2(\D{R}).
\end{split}
\end{equation}
where
\[
q(y,t)= \frac{\varkappa}{4}\frac{w_{xx}(x,t)}{w(x,t)^2}-\frac{w(x,t)-\varkappa}{4 w(x,t)}-\frac{5\varkappa}{16}\frac{w_x(x,t)^2}{w(x,t)^3}.
\]
From our assumption \eqref{decay} it follows that $q(y,t)\in L^1(\D{R},(1+|y|)^{l+1}\dd y)$.

\begin{lemma}
There exist two Jost solutions $\psi_\pm(k,x,t)$ which solve the differential equation
\begin{equation}
H(t) \psi_\pm(k,x,t)=\frac{1}{\varkappa}\left(\frac{1}{4}+k^2\right)\psi_\pm(k,x,t), \qquad \Im(k)\geq 0,
\end{equation}
and
\begin{equation}
\lim_{x\to\pm\infty}\ee^{\mp\ii kx} \psi_{\pm}(k,x,t) = 1.
\end{equation}
Both $\psi_\pm(k,x,t)$ are analytic for $\Im(k) > 0$ and continuous
for $\Im(k)\geq 0$. For large $k$ we have
\begin{equation}\label{eq:psiasym}
\psi_\pm(k,x,t) =\ee^{\pm\ii k (y + \frac{1\mp 1}{2}H_{-1}(w))}
\frac{\varkappa^{1/4}}{w(x,t)^{1/4}}
\left(1\mp \int_y^{\pm\infty}\!q(r,t)\dd r\,\frac{1}{2\ii k}+\ord\Bigl(\frac{1}{k^2}\Bigr)\right)
\end{equation}
as $k\to\infty$, where
\begin{equation}
H_{-1}(w)=\int_{\D{R}}\biggl(\sqrt{\tfrac{w(r)}{\varkappa}}-1\biggr)\dd r
\end{equation}
is a conserved quantity of the CH equation.
\end{lemma}

\begin{proof}
This is immediate from the corresponding results for \eqref{schroeop} (cf., e.g., \cite{dt} or \cite{mar}) by virtue of
our Liouville transform \eqref{liovtrf}. Just observe
\[
\psi_\pm(k,x,t) = \ee^{-\ii k \frac{1\mp 1}{2}H_{-1}(w)}
\frac{\varkappa^{1/4}}{w(x,t)^{1/4}} \tilde{\psi}_\pm(k,y,t),
\]
where $\tilde{\psi}_\pm(k,y,t)$ are the Jost solutions of \eqref{schroeop}.
\end{proof}

Furthermore, one has the scattering relations
\begin{equation} \label{relscat}
T(k) \psi_\mp(k,x,t) =  \overline{\psi_\pm(k,x,t)} +
R_\pm(k,t) \psi_\pm(k,x,t),  \qquad k \in \D{R},
\end{equation}
where $T(k)$, $R_\pm(k,t)$ are the transmission, resp.\ reflection coefficients. We have symmetry relations 
\[
R_\pm(-k,t)=\overline{R_\pm(k,t)}\ \text{ and }\ T(-k)=\overline{T(k)}.
\]
Note also that if $\tilde{T}(k)$, $\tilde{R}_\pm(k,t)$
are the corresponding quantities for $\tilde{H}(t)$, then $T(k)= \ee^{\ii k H_{-1}(w)} \tilde{T}(k)$, $R_+(k,t)= \tilde{R}_+(k,t)$, and
$R_-(k,t)= \ee^{2\ii k H_{-1}(w)} \tilde{R}_-(k,t)$ and hence all results known for \eqref{schroeop} readily apply in our situation.
In particular, they have the following well-known properties:

\begin{lemma}\label{lem:tr_coef}
The transmission coefficient $T(k)$ is meromorphic for $\Im(k) > 0$
with simple poles at $\ii\kappa _1, \dots,\ii\kappa_N$, where as above $\kappa_j=\sqrt{\frac{1}{4} - \varkappa \lambda_j}\in(0,\frac{1}{2})$,
and is continuous up to the real line. Asymptotically we have
\begin{equation}\label{eq:asymT}
T(k)=\ee^{\ii k H_{-1}(w)} (1+\ord(k^{-1})).
\end{equation}
The residues of $T(k)$ are given by
\begin{equation}\label{eq:resT}
\Res_{\ii\kappa_j} T(k) = \ii\mu _j(t) \gamma _{+,j}(t)^2 = \ii\mu _j \gamma_{+,j}^2,
\end{equation}
where
\begin{equation}
\gamma_{+,j}(t)^{-2} = \varkappa^{-1} \int_{\D{R}} \psi _+(\ii\kappa_j,r,t)^2 w(r,t)\dd r
\end{equation}
and $\psi_+ (\ii\kappa_j,x,t) = \mu_j(t) \psi_-(\ii\kappa_j,x,t)$.

Moreover,
\begin{equation} \label{reltrpm}
T(k) \overline{R_+(k,t)} + \overline{T(k)} R_-(k,t)=0, \qquad |T(k)|^2 + |R_\pm(k,t)|^2=1.
\end{equation}
\end{lemma}

Note that one reflection coefficient, say $R(k,t)=R_+(k,t)$, and one set of
norming constants, say $\gamma_j(t):=\gamma_{+,j}(t)$, suffices.

The time dependence is given by (see \cite{co}):

\begin{lemma}
The time evolutions of the quantities $R(k,t)$ and $\gamma_j(t)$ are given by,
\begin{align}
R(k,t) &= R(k) \ee^{- \ii\frac{\varkappa k}{1/4+k^2} t},\\
\gamma_j(t) &= \gamma_j \ee^{\frac{\varkappa \kappa_j/2}{1/4-\kappa_j^2} t}
\end{align}
where $R(k)=R(k,0)$ and $\gamma_j=\gamma_j(0)$.
\end{lemma}

\subsection*{Vector Riemann\textendash Hilbert problem}
We will set up a vector Riemann--Hilbert problem as follows. Let $m(k,x,t)=\bigl(m_1(k,x,t)\ \ m_2(k,x,t)\bigr)$ be defined by
\begin{equation}\label{defm}
\begin{cases}
(\frac{w(x,t)}{\varkappa})^{1/4} \begin{pmatrix} T(k) \psi_-(k,x,t) \ee^{\ii k y}  & \psi_+(k,x,t) \ee^{-\ii k y} \end{pmatrix},
& \Im(k) > 0,\\
(\frac{w(x,t)}{\varkappa})^{1/4} \begin{pmatrix} \psi_+(-k,x,t) \ee^{\ii k y} & T(-k) \psi_-(-k,x,t) \ee^{-\ii k y} \end{pmatrix},
& \Im(k) < 0.
\end{cases}
\end{equation}
We are interested in the jump condition of $m(k,x,t)$ on the real $k$-axis (oriented
from negative to positive).
To formulate our jump condition we use the following convention:
when representing functions on $\D{R}$, the lower subscript denotes
the non-tangential limit from different sides.
By $m_+(k)$ we denote the limit from above and by $m_-(k)$ the one from below.
Using the notation above implicitly assumes that these limits exist in the sense that
$m(k)$ extends to a continuous function on the real axis.
In general, for an oriented contour $\Sigma$, $m_+(k)$ (resp.\ $m_-(k)$) will denote the limit
of $m(\kappa)$ as $\kappa\to k$ from the positive (resp.\ negative) side of $\Sigma$. Here
the positive (resp.\ negative) side is the one which lies to the left (resp.\ right) as one traverses the contour in the
direction of the orientation.

\begin{theorem}[vector RH-problem]\label{thm:vecrhp}
Let $\mathcal{S}_+(H(0))=\{R(k),\: (\kappa_j, \gamma_j), \: j=1,\dots,N\}$ be
the right scattering data of the operator $H(0)$ associated with the initial data $w(x,0)$. Then $m(k)\equiv m(k,x,t)$ defined in \eqref{defm}
is a solution of the following vector Riemann--Hilbert problem.
Find a function $m(k)$ which satisfies:
\begin{enumerate}[\rm(i)]
\item 
The analyticity condition: 
\begin{enumerate}[]
\item 
$m(k)$ is meromorphic away from the real axis with simple poles at $\pm\ii\kappa_j$.
\end{enumerate}
\item 
The jump condition, for $k\in\D{R}$:
\begin{equation} \label{eq:jumpcond}
\aligned
m_+(k)&=m_-(k) v(k), \\
v(k)&=\begin{pmatrix}
1-|R(k)|^2 & -\overline{R(k)}\ee^{-t\Phi(k)} \\
R(k)\ee^{t\Phi(k)} & 1
\end{pmatrix}.
\endaligned
\end{equation}
\item
The pole conditions, for $j=1,\dots,N$:
\begin{equation}\label{eq:polecond}
\Res_{\ii\kappa_j} m(k) = \lim_{k\to\ii\kappa_j} m(k)
\begin{pmatrix} 0 & 0\\ \ii\gamma_j^2 \ee^{t\Phi(\ii\kappa_j)}  & 0 \end{pmatrix}.
\end{equation}
In \emph{(ii)} and \emph{(iii)} the phase is given by
\begin{equation}
\Phi(k)= -\ii\frac{\varkappa k}{\tfrac{1}{4}+k^2} +2\ii k \frac {y}{t}.
\end{equation}
\item
The symmetry condition
\begin{equation} \label{eq:symcond}
m(-k) = m(k) \begin{pmatrix}0&1\\1&0\end{pmatrix}.
\end{equation}
\item
The normalization
\begin{equation}\label{eq:normcond}
\lim_{k\to\infty}m(k)= (1\quad 1).
\end{equation}
\end{enumerate}
\end{theorem}

\begin{remarks*}
\begin{enumerate}[(a)]
\item 
Note that $\det v(k)\equiv 1$ and $v(-k)=\sigma_1v(k)^{-1}\sigma_1$ with $\sigma_1=\left(\begin{smallmatrix}0&1\\1&0\end{smallmatrix}\right)$. 
\item
Note also that (iii) and (iv) imply the pole conditions, for $j=1,\dots,N$:
\[
\Res_{-\ii\kappa_j}m(k) = \lim_{k\to -\ii\kappa_j} m(k)
\begin{pmatrix} 0 & - \ii\gamma_j^2 \ee^{t\Phi(\ii\kappa_j)} \\ 0 & 0 \end{pmatrix}.
\]
\end{enumerate}
\end{remarks*}

\begin{proof}
The jump condition \eqref{eq:jumpcond} is a simple calculation using the scattering relations
\eqref{relscat} plus \eqref{reltrpm}. The pole conditions follow since $T(k)$ is meromorphic for $\Im(k) > 0$
with simple poles at $\ii\kappa_j$ and residues given by \eqref{eq:resT}.
The symmetry condition holds by construction and the normalization \eqref{eq:normcond}
is immediate from \eqref{eq:psiasym} and
\begin{equation}
T(k)=\ee^{\ii k H_{-1}(w)} \left(1+ \int_{-\infty}^{+\infty} q(r,t)\dd r (2\ii k)^{-1} + \ord(k^{-2}) \right)
\end{equation}
which implies
\begin{equation}\label{eq:asym_infty}
m(k,x,t)=\begin{pmatrix}1&1\end{pmatrix}+Q_+(y,t)\frac{1}{2\ii k}\begin{pmatrix}1&-1\end{pmatrix}+\ord(k^{-2}),
\end{equation}
where $Q_+(y,t)= \int_y^{+\infty} q(r,t)\dd r$.
\end{proof}

Observe that the pole condition at $\ii\kappa_j$ is sufficient since the one at $-\ii\kappa_j$ follows
by symmetry. Hence, the Riemann--Hilbert problem for the Camassa--Holm equation is, for given
scattering data $\mathcal{S}_+$, to find a sectionally meromorphic vector function $m(k)$
satisfying \eqref{eq:jumpcond}--\eqref{eq:normcond}. We will show that the solution given
in the above theorem is in fact the only one in Corollary~\ref{cor:unique} below.
Moreover, it should be pointed out that except for the phase, this Riemann--Hilbert
problem is identical to the one for the Korteweg--de Vries equation (cf.\ \cite[Thm.~2.3]{gt}).

Next we note the following useful asymptotics

\begin{lemma}\label{lem:asymp}
The function $m(k,x,t)$ defined in \eqref{defm} satisfies
\begin{equation}\label{eq:asymy}
\frac{m_1(\tfrac{\ii}{2},x,t)}{m_2(\tfrac{\ii}{2},x,t)} = \ee^{x-y},
\end{equation}
and
\begin{equation}\label{eq:asym}
m_1(k,x,t) m_2(k,x,t) =  \sqrt{\frac{w(x,t)}{\varkappa}}\Bigl( 1 + \tfrac{2\ii}{\varkappa} u(x,t) \left(k-\tfrac{\ii}{2}\right)
+\ord(k-\tfrac{\ii}{2})^2\Bigr).
\end{equation}
\end{lemma}

\begin{proof}
The Jost solutions admit the representation $\psi_\pm(k,x,t)=\ee^{\pm\ii kx}g_\pm(k,x,t)$, where
$g_\pm(k,x,t)$ are the solutions of the integral equations
\begin{equation}\label{int_eqs_02}
g_\pm(k,x,t)=1\pm\frac{k^2+\tfrac{1}{4}}{2\ii\varkappa k}\int_x^{\pm\infty}(1-\ee^{\mp 2\ii k(x-r)})g_\pm(k,r,t)(w(r,t)-\varkappa)\dd r.
\end{equation}
Since $w$ satisfies \eqref{decay}, the solution of \eqref{int_eqs_02} exists and is unique (see, for example, \cite{mar}).
Moreover, $g_\pm$ is analytic for $\Im(k)>0$.

Since $k^2+\tfrac{1}{4} = (k-\tfrac{\ii}{2})(k+\tfrac{\ii}{2})$, we get
\begin{equation}\label{g_ass_01}
g_\pm(k,x,t)=1\pm \frac{\ii}{\varkappa}(k-\tfrac{\ii}{2})F_\pm(x,t)+\ord(k-\tfrac{\ii}{2})^2,\quad k\to \ii /2,
\end{equation}
where
\begin{equation}\label{g_ass_02}
F_\pm(x,t)=\int_x^{\pm\infty} (\ee^{\pm (x-r)}-1)(w(r,t)-\varkappa)\dd r.
\end{equation}
Moreover, differentiating with respect to $x$, we see
\begin{equation}\label{g_ass_03}
g_\pm'(k,x,t)=\pm \frac{\ii}{\varkappa}(k-\tfrac{\ii}{2})F_\pm'(x,t)+\ord(k-\tfrac{\ii}{2})^2,
\end{equation}
with
\begin{equation}\label{g_ass_02'}
F_\pm'(x,t)= \pm\int_x^{\pm\infty} \ee^{\pm (x-r)}(w(r,t)-\varkappa)\dd r.
\end{equation}
Using
\begin{equation}\label{momentum_inv}
u(x,t)= \left(1-\partial_x^2\right)^{-1} (w(x,t)-\varkappa) =
\frac{1}{2} \int_{\D{R}} \ee^{-|x-r|}  (w(r,t)-\varkappa)\dd r,
\end{equation}
we thus obtain
\begin{align}\notag
\psi_+(k,x,t)\psi_-(k,x,t)
&= 1+\frac{\ii}{\varkappa}(F_+(x,t)-F_-(x,t))(k-\tfrac{\ii}{2})+\ord(k-\tfrac{\ii}{2})^2 \\ \label{g_ass_04}
&= 1+\frac{\ii}{\varkappa}(2u(x,t)-H_0(u))(k-\tfrac{\ii}{2})+\ord(k-\tfrac{\ii}{2})^2,
\end{align}
where
\begin{equation}
H_0(u) = \int_{\D{R}} (w(x,t)-\varkappa)\dd x = \int_{\D{R}} u(x,t)\dd x
\end{equation}
is a conserved quantity of the CH equation.

Furthermore, straightforward calculations show that
\begin{align} \label{g_ass_05}
T(k)^{-1} 
&= \frac{\wronsk(\psi_-,\psi_+)}{2\ii k}\notag\\
&=1+\frac{\ii}{\varkappa}\bigl(F_+(x,t)-F_-(x,t)- F'_+(x,t) - F_-'(x,t)\bigr) (k-\tfrac{\ii}{2})+\ord(k-\tfrac{\ii}{2})^2\notag\\
&= 1-\frac{\ii}{\varkappa} H_0(u) (k-\tfrac{\ii}{2})+\ord(k-\tfrac{\ii}{2})^2,
\end{align}
where $\wronsk(f,g)=fg'-f'g$ is the usual Wronskian.
Therefore,
\begin{equation}\label{g_ass_06}
T(k) =  1 + \frac{\ii}{\varkappa} H_0(u) \left(k-\tfrac{\ii}{2}\right) +O\!\left(k-\tfrac{\ii}{2}\right)^2.
\end{equation}
Substituting \eqref{g_ass_01} and \eqref{g_ass_06} into \eqref{defm}, we arrive at \eqref{eq:asymy}.
Substituting \eqref{g_ass_04} and \eqref{g_ass_06} into \eqref{defm}, we obtain \eqref{eq:asym}.
\end{proof}

\subsection*{Regular Riemann\textendash Hilbert problem}
For our further analysis it will be convenient to rewrite the pole condition as a jump
condition and hence turn our meromorphic Riemann--Hilbert problem into a holomorphic Riemann--Hilbert problem.

Choose $\varepsilon$ so small that the discs $\abs{k-\ii\kappa_j}<\varepsilon$ lie inside the upper half plane and do not intersect. Then redefine $m(k)$ in a neighborhood of $\ii\kappa_j$, resp.~$- \ii\kappa_j$ according to
\begin{equation}\label{eq:redefm}
\hat m(k) = 
\begin{cases} 
m(k) 
\begin{pmatrix} 
1 & 0 \\
-\frac{\ii\gamma_j^2\ee^{t\Phi(\ii\kappa_j)} }{k- \ii\kappa_j} & 1 
\end{pmatrix},
&\abs{k-\ii\kappa_j}<\varepsilon,\\
m(k) 
\begin{pmatrix} 
1 & \frac{\ii\gamma_j^2 \ee^{t\Phi(\ii\kappa_j)} }{k+ \ii\kappa_j} \\
0 & 1 
\end{pmatrix},
&\abs{k+\ii\kappa_j}<\varepsilon,\\
m(k),&\text{else}.
\end{cases}
\end{equation}
Note that we redefined $m(k)$ such that it respects our symmetry \eqref{eq:symcond}. 

Then a straightforward calculation using
$\Res_{\ii\kappa}m(k)=\lim_{k\to\ii\kappa}(k-\ii\kappa)m(k)$ shows:

\begin{lemma}[regular RH-problem]\label{lem:holrhp}
Let $C_j$ be the circle $\abs{k-\ii\kappa_j}=\varepsilon$ with $\varepsilon>0$ as above, $1\leq j\leq N$. Let $\hat m(k)$ be defined as in \eqref{eq:redefm}. Then $\hat m(k)$ is a solution of the following vector Riemann\textendash Hilbert problem. Find a function $\hat m(k)$ which satisfies:
\begin{enumerate}[\rm(i)]
\item
$\hat m(k)$ is holomorphic away from the real axis and from the circles $C_j$ and $\bar C_j$, for $j=1,\dots,N$. 
\item
The jump condition \eqref{eq:jumpcond} across the real axis.
\item
The additional jump conditions across the circles $C_j$, $\bar C_j$, for $j=1,\dots,N$:
\begin{equation} \label{eq:jumpcond2}
\aligned
\hat m_+(k) &=\hat m_-(k) \begin{pmatrix} 1 & 0 \\
-\frac{\ii\gamma_j^2 \ee^{t\Phi(\ii\kappa_j)}}{k-\ii\kappa_j} & 1 \end{pmatrix},\quad k\in C_j,\\
\hat m_+(k) &=\hat m_-(k) 
\begin{pmatrix} 1 & -\frac{\ii\gamma_j^2 \ee^{t\Phi(\ii\kappa_j)}}{k+ \ii\kappa_j} \\
0 & 1
\end{pmatrix},\quad k\in\bar C_j,
\endaligned
\end{equation}
where $C_j$ is oriented counterclockwise and $\bar C_j$ is oriented clockwise.
\item
The symmetry condition \eqref{eq:symcond}.
\item
The normalization condition \eqref{eq:normcond}.
\end{enumerate}
\end{lemma}

\begin{figure}[ht]
\begin{pspicture}(-7.5,-3.3)(7.5,3.7)
\psset{unit=.8}
\psline[linewidth=1pt,arrowsize=3pt 2,arrowlength=1.3]{->}(-8,0)(-3.7,0)
\psline[linewidth=1pt,arrowsize=3pt 2,arrowlength=1.3]{->}(-3.8,0)(4.1,0)
\psline[linewidth=1pt]{-}(4,0)(8,0)
\psline[linewidth=.15pt,linestyle=dashed](0,-4.8)(0,4.8)
\psarc[linewidth=.8pt,arrowsize=3pt 2,arrowlength=1.3]{->}(0,.7){.3}{-95}{100}
\psarc[linewidth=.8pt,arrowsize=3pt 2,arrowlength=1.3]{->}(0,.7){.3}{-270}{-80}
\psarc[linewidth=.8pt,arrowsize=3pt 2,arrowlength=1.3]{->}(0,2){.3}{-95}{100}
\psarc[linewidth=.8pt,arrowsize=3pt 2,arrowlength=1.3]{->}(0,2){.3}{-270}{-80}
\psarc[linewidth=.8pt,arrowsize=3pt 2,arrowlength=1.3]{->}(0,3.4){.3}{-95}{100}
\psarc[linewidth=.8pt,arrowsize=3pt 2,arrowlength=1.3]{->}(0,3.4){.3}{-270}{-80}
\psarc[linewidth=.8pt,arrowsize=3pt 2,arrowlength=1.3]{<-}(0,-.7){.3}{-95}{100}
\psarc[linewidth=.8pt,arrowsize=3pt 2,arrowlength=1.3]{<-}(0,-.7){.3}{-270}{-80}
\psarc[linewidth=.8pt,arrowsize=3pt 2,arrowlength=1.3]{<-}(0,-2){.3}{-95}{100}
\psarc[linewidth=.8pt,arrowsize=3pt 2,arrowlength=1.3]{<-}(0,-2){.3}{-270}{-80}
\psarc[linewidth=.8pt,arrowsize=3pt 2,arrowlength=1.3]{<-}(0,-3.4){.3}{-95}{100}
\psarc[linewidth=.8pt,arrowsize=3pt 2,arrowlength=1.3]{<-}(0,-3.4){.3}{-270}{-80}
\rput(7.6,-.3){$\D{R}$}
\rput(-.3,4.4){$\frac{\ii}{2}$}
\rput(-.4,-4.4){$-\frac{\ii}{2}$}
\put(.4,.6){\scriptsize$C_N$}
\put(0,.7){\circle*{.1}}
\put(.4,1.9){\scriptsize$C_j$}
\put(0,2){\circle*{.1}}
\put(.4,3.3){\scriptsize$C_1$}
\put(0,3.4){\circle*{.1}}
\put(0,4.4){\circle*{.1}}
\put(.4,-.8){\scriptsize$\bar C_N$}
\put(0,-.7){\circle*{.1}}
\put(.4,-2.1){\scriptsize$\bar C_j$}
\put(0,-2){\circle*{.1}}
\put(.4,-3.5){\scriptsize$\bar C_1$}
\put(0,-3.4){\circle*{.1}}
\put(0,-4.4){\circle*{.1}}
\end{pspicture}
\caption{Contour of the regular RH problem}
\label{fig:contour}
\end{figure}
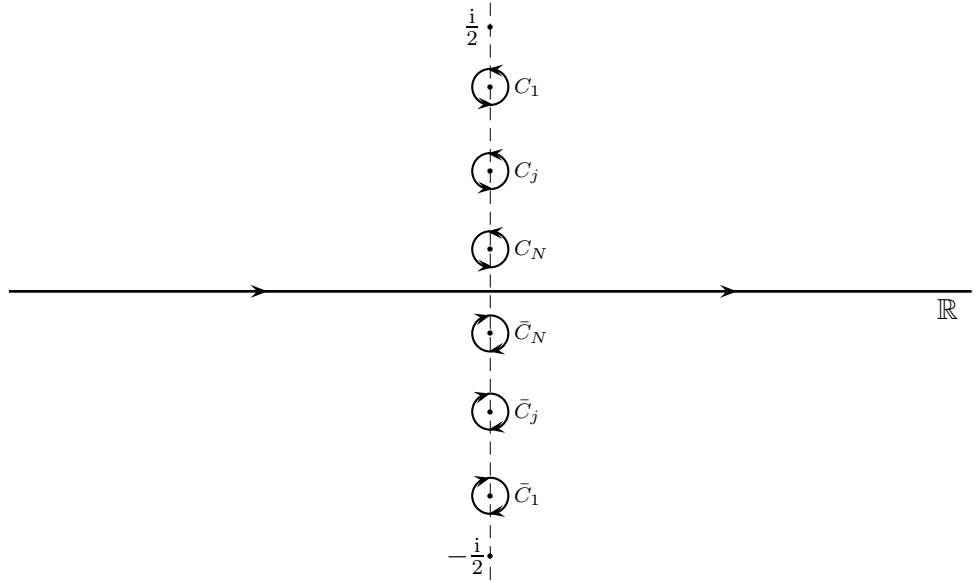

\subsection*{Uniqueness result}
Next we turn to uniqueness of the solution of this vector Riemann--Hilbert problem. This will also explain the reason for our symmetry condition. We begin by observing that if there is a point $k_1\in\D{C}$, such that $m(k_1)=\begin{pmatrix}0&0\end{pmatrix}$, then $n(k)=\frac{1}{k-k_1} m(k)$ is a solution of the associated vanishing Riemann--Hilbert problem, i.e., it satisfies the same jump and pole conditions as $m(k)$ but the normalization now reads $\lim_{\kappa\to\infty} m(\ii\kappa) = \begin{pmatrix}0&0\end{pmatrix}$. In particular, there is a whole family of solutions $m(k)+ \vartheta\, n(k)$ for any $\vartheta\in\D{C}$. However, these solutions will clearly violate the symmetry condition unless $\vartheta=0$! Hence, without the symmetry condition, the solution of our vector Riemann--Hilbert problem will not be unique in such a situation. Moreover, a look at the one-soliton solution verifies that this case indeed can happen.

\begin{lemma}[one-soliton solution]\label{lem:singlesoliton}
Suppose that the reflection coefficient vanishes, \emph{i.e.}, $R(k,t)\equiv 0$ and that there is only one eigenvalue $\kappa\in(0,\frac{1}{2})$, with the norming constant $\gamma(t)$. Then the unique solution of the Riemann--Hilbert problem \eqref{eq:jumpcond}--\eqref{eq:normcond}
is given by
\begin{align}\label{eq:oss}
m_0(k) &= \begin{pmatrix} f(k) & f(-k) \end{pmatrix} \\
\notag 
f(k) &= \frac{1}{1+\alpha}
\left(1+ \alpha \frac{k+\ii\kappa}{k-\ii\kappa}\right), \\
\notag
\alpha&= \frac{\gamma^2}{2\kappa} \ee^{t\Phi(\ii\kappa)}.
\end{align}
In particular,
\begin{align}\label{eq:uoss}
u(x,t) &= \frac{32\varkappa\kappa^2}{(1-4\kappa^2)^2}\ \alpha(y,t)
\biggl((1+\alpha(y,t))^2+ \frac{16\kappa^2}{1-4\kappa^2}\ \alpha(y,t)\biggr)^{-1},\\          \label{eq:woss}
w(x,t) &= \varkappa\left(1+\frac{16\kappa^2}{1-4\kappa^2}\frac{\alpha(y,t)}{(1+\alpha(y,t))^2}\right)^2,
\end{align}
where
\begin{equation}
x = y + \log\frac{1+\alpha(y,t) \frac{1+2\kappa}{1-2\kappa}}{1+\alpha(y,t) \frac{1-2\kappa}{1+2\kappa}}\,.
\end{equation}
\end{lemma}

\begin{proof}
By assumption the reflection coefficient vanishes and so the jump along the real axis
disappears. Therefore and by the symmetry condition, we know that the solution is of the form
$m_0(k) = \begin{pmatrix} f(k) & f(-k) \end{pmatrix}$ where $f(k)$ is meromorphic. Furthermore the function
$f(k)$ has only a simple pole at $\ii\kappa$, so that we can make the ansatz
$f(k)=C+D\frac{k+\ii\kappa}{k-\ii\kappa}$. Then the constants $C$ and $D$ are uniquely determined by the
pole conditions and the normalization. Formul\ae\ \eqref{eq:uoss}-\eqref{eq:woss} are obtained applying Lemma \ref{lem:asymp}, formula \eqref{eq:asym}.
\end{proof}

In fact, observe that $f(k_1)=f(-k_1)=0$ if and only if $k_1=0$ and $2\kappa=\gamma^2\ee^{t\Phi (\ii\kappa)} $.
Furthermore, even in the general case $m(k_1)=\begin{pmatrix}0&0\end{pmatrix}$ can only occur at $k_1=0$ as the
following lemma shows.

\begin{lemma}\label{lem:resonant}
If $m(k_1) = \begin{pmatrix}0&0\end{pmatrix}$ for $m$ defined as in \eqref{defm}, then $k_1  = 0$. Moreover,
the zero of at least one component is simple in this case.
\end{lemma}

\begin{proof}
By \eqref{defm} the condition $m(k_1) = \begin{pmatrix}0&0\end{pmatrix}$ implies that the Jost solutions $\psi_-(k,x)$ and
$\psi_+(k,x)$ are linearly dependent or that the transmission coefficient $T(k_1)=0$. This can only happen at the band edge, $k_1 = 0$
or at an eigenvalue $k_1=\ii\kappa_j$.

We begin with the case $k_1=\ii\kappa_j$. In this case the $\psi_-(k,x)$ and
$\psi_+(k,x)$ are linearly dependent. Moreover, $T(\,\cdot\,)$ has a simple pole at $k=k_1$ since the derivative of the Wronskian
$\wronsk(k)= \psi_+(k,x)\psi_-'(k,x)-\psi_+'(k,x)\psi_-(k,x)$ does not vanish by the well-known formula
\[
\frac{\dd}{\dd k}\wronsk(k) |_{k=k_1} = - \frac{2k_1}{\varkappa}\int_{\D{R}} \psi_+(k_1,x)  \psi_-(k_1,x) w(x)\dd x \neq 0
\]
(cf.~Lemma \ref{lem:tr_coef}). The diagonal Green's function $g(z,x)=\wronsk(k)^{-1} \psi_+(k,x) \psi_-(k,x)$ is
Herglotz as a function of $z=-k^2$ and hence can have at most a simple zero at $z=-k_1^2$.
Since $z\mapsto-k^2$ is conformal away from $z=0$ the same is true as a function of $k$. Hence, if
$\psi_+(\ii\kappa_j,x) = \psi_-(\ii\kappa_j,x) =0$, both can have at most a simple zero at $k=\ii\kappa_j$.
But $T(k)$ has a simple pole at $\ii\kappa_j$ and hence $T(k) \psi_-(k,x)$ cannot
vanish at $k=\ii\kappa_j$, a contradiction.

It remains to show that one zero is simple in the case $k_1=0$. In fact,
one can show that 
\[
\frac{\dd}{\dd k}\wronsk(k) |_{k=k_1} \neq 0
\]
in this case as follows:
first of all note that $\dot{\psi}_\pm(k)$ (where the dot denotes the derivative with respect to
$k$) again solves
\[
H\dot{\psi}_\pm(k_1) = \frac{1}{\varkappa}\left(\frac14+k_1^2\right) \dot{\psi}_\pm(k_1)
\]
if $k_1=0$. Moreover, by
$\wronsk(k_1)=0$ we have $\psi_+(k_1) = c\, \psi_-(k_1)$ for some constant $c$ (independent of $x$).
Thus we can compute
\begin{align*}
\dot{\wronsk}(k_1) 
&=\wronsk(\dot{\psi}_+(k_1),\psi_-(k_1))+\wronsk(\psi_+(k_1),\dot{\psi}_-(k_1))\\
&= c^{-1}\wronsk(\dot{\psi}_+(k_1),\psi_+(k_1))+c\wronsk(\psi_-(k_1),\dot{\psi}_-(k_1))
\end{align*}
by letting $x\to+\infty$ for the first and $x\to-\infty$ for the second Wronskian (in which case we can
replace $\psi_\pm(k)$ by $\ee^{\pm\ii k x}$),
which gives
\[
\dot{\wronsk}(k_1)=-\ii(c+c^{-1}).
\]
Hence the Wronskian has a simple zero. But if both functions had more than
simple zeros, so would the Wronskian, a contradiction.
\end{proof}

By \cite[Theorem~3.2]{gt} we obtain

\begin{corollary}\label{cor:unique}
The function $m(k,x,t)$ defined in \eqref{defm} is the only solution of the
vector Riemann--Hilbert problem \eqref{eq:jumpcond}--\eqref{eq:normcond}.
\end{corollary}

\section{Conjugation and Deformation}              \label{sec:condef}

This section demonstrates how to conjugate our Riemann--Hilbert problem (with respect to the augmented contour) and how to deform our jump
contour, such that the jumps will be exponentially close to the identity away from the stationary
phase points. Throughout this and the following section, we will assume that the $R(k)$ has an analytic
extension to a small neighborhood of the real axis. This is for example the case if we assume that our solution is exponentially decaying. This assumption can then be removed using analytic approximation.

For easy reference we note the following result:

\begin{lemma}[conjugation]\label{lem:conjug}
Let $\widetilde{\Sigma}$ be a part of some contour $\Sigma$. Let $D$ be a matrix of the form
\begin{equation}
D(k) = 
\begin{pmatrix} 
d(k)^{-1} & 0 \\ 0 & d(k) 
\end{pmatrix},
\end{equation}
where $d\colon\D{C}\setminus\widetilde{\Sigma}\to\D{C}$ is a sectionally analytic function. Set
\begin{equation}
\tilde{m}(k)=m(k)D(k),
\end{equation}
then the jump matrix transforms according to
\begin{equation}
\tilde v(k)=D_-(k)^{-1} v(k) D_+(k).
\end{equation}
If $d$ satisfies $d(-k) = d(k)^{-1}$ and $\lim_{k\to\infty}d(k)=1$, then the transformation 
\[
\tilde{m}(k) = m(k) D(k)
\]
respects our symmetry condition, that is, $\tilde{m}(k)$ satisfies \eqref{eq:symcond} if and only if $m(k)$ does, and our normalization condition.
\end{lemma}

In particular, we obtain
\begin{equation}
\tilde v(k)= 
\begin{cases}
\begin{pmatrix} v_{11} & v_{12} d^{2} \\ v_{21} d^{-2}  & v_{22} \end{pmatrix},
&k\in\Sigma\setminus\widetilde{\Sigma},\\[4mm]
\begin{pmatrix} \frac{d_-}{d_+} v_{11} & v_{12} d_+ d_- \\
v_{21} d_+^{-1} d_-^{-1}  & \frac{d_+}{d_-} v_{22} \end{pmatrix},
&k\in\widetilde{\Sigma}.
\end{cases}
\end{equation}
In order to analyse the regular vector RH problem from Lemma \ref{lem:holrhp} there are two cases to distinguish. 
\begin{enumerate}[(a)]
\item
If $\Phi(\ii\kappa_j)<0$
then the corresponding jump matrix \eqref{eq:jumpcond2} is exponentially close to the identity
as $t\to+\infty$ and there is nothing to do.
\item
Otherwise we use conjugation to turn the jumps into one with exponentially decaying
off-diagonal entries.
\end{enumerate}
It turns out that we will have to handle the jumps across $C_j$ and $\bar C_j$ in one step in order to preserve symmetry and in order to not add additional singularities elsewhere.

\begin{lemma}\label{lem:twopolesinc}
Assume that the Riemann--Hilbert problem for $m$ has jump conditions near $\ii\kappa$ and
$-\ii\kappa$ given by
\begin{equation}
m_+(k)=
\begin{cases}
m_-(k)\begin{pmatrix}1&0\\ -\frac{\ii\gamma^2}{k-\ii\kappa}&1\end{pmatrix}, 
&\abs{k-\ii\kappa}=\varepsilon, \\[4mm]
m_-(k)\begin{pmatrix}1& -\frac{\ii\gamma^2}{k+\ii\kappa}\\0&1\end{pmatrix}, 
&\abs{k+\ii\kappa}=\varepsilon.
\end{cases}
\end{equation}
Then this Riemann--Hilbert problem is equivalent to a Riemann--Hilbert problem for $\tilde{m}=mD$ which
has jump conditions near $\ii\kappa$ and $-\ii\kappa$ given by
\begin{equation}
\tilde{m}_+(k)= 
\begin{cases}
\tilde{m}_-(k)\begin{pmatrix}1& -\frac{(k+\ii\kappa)^2}{\ii\gamma^2(k-\ii\kappa)}\\ 0 &1\end{pmatrix},
&\abs{k-\ii\kappa}=\varepsilon, \\[4mm]
\tilde{m}_-(k)\begin{pmatrix}1& 0 \\ -\frac{(k-\ii\kappa)^2}{\ii\gamma^2(k+\ii\kappa)}&1\end{pmatrix},
&\abs{k+\ii\kappa}=\varepsilon,
\end{cases}
\end{equation}
and all remaining data conjugated by
\begin{equation}
D(k) = 
\begin{cases}
\begin{pmatrix} 1 & -\frac{k-\ii\kappa}{\ii\gamma^2} \\  \frac{\ii\gamma^2}{k-\ii\kappa} & 0 \end{pmatrix}
\begin{pmatrix} \frac{k-\ii\kappa}{k+\ii\kappa} & 0 \\ 0 & \frac{k+\ii\kappa}{k-\ii\kappa} \end{pmatrix}, &  \left\vert k-\ii\kappa \right\vert <\varepsilon, \\
\begin{pmatrix} 0 & -\frac{\ii\gamma^2}{k+\ii\kappa} \\ \frac{k+\ii\kappa}{\ii\gamma^2} & 1 \end{pmatrix} 
\begin{pmatrix} \frac{k-\ii\kappa}{k+\ii\kappa} & 0 \\ 0 & \frac{k+\ii\kappa}{k-\ii\kappa} \end{pmatrix}, & \left\vert k+\ii\kappa \right\vert <\varepsilon, \\ 
\begin{pmatrix} \frac{k-\ii\kappa}{k+\ii\kappa} & 0 \\ 0 & \frac{k+\ii\kappa}{k-\ii\kappa} \end{pmatrix}, & \text{else}.
\end{cases}
\end{equation}
\end{lemma}

The jump along the real axis is of oscillatory type and our aim is to apply
a contour deformation following \cite{dz} such that all jumps will be moved into regions where the oscillatory terms
will decay exponentially. Since the jump matrix $v$ contains both $\exp(t\Phi)$ and
$\exp(-t\Phi)$ we need to separate them in order to be able to move them to different regions
of the complex plane.

We recall that the phase of the associated Riemann--Hilbert problem is given by
\begin{equation} \label{eq:Phi}
\Phi(k)= -\ii\frac{\varkappa k}{\tfrac{1}{4}+k^2}+2\ii k\frac{y}{t}.
\end{equation}
Let
\begin{equation}  \label{eq:slope}
c\equiv c(y,t)=\frac{y}{\varkappa t}\,.
\end{equation}
The stationary phase points, i.e., $\Phi'(k)=0$, are given by $\pm k_0$ and $\pm k_1$, where
\begin{equation}       \label{eq:stationary.phase.points}
\begin{split}
k_0&\equiv k_0(c) = \frac{1}{2} \sqrt{-\frac{1+c-\sqrt{1+4c}}{c}},\\
k_1&\equiv k_1(c) = \frac{1}{2} \sqrt{-\frac{1+c+\sqrt{1+4c}}{c}}.
\end{split}
\end{equation}
There are four cases to distinguish:
\begin{enumerate}[(i)]
\item 
$2\varkappa<\frac{y}{t}$, i.e.\ $c>2$.
In this case 
\begin{enumerate}[$\triangleright$]
\item
$k_0,k_1\not\in\D{R}$,
\item
$\Re(\Phi)<0$ for $\Im(k)>0$ near $\D{R} \cup \ii(0,\kappa_0)$
and 
\item
$\Re(\Phi)>0$ for $\Im(k)>0$ near $\ii(\kappa_0,\tfrac{1}{2})$, where
\begin{equation}
\kappa_0=\tfrac{1}{2}\sqrt{1-\tfrac{2}{c}}.
\end{equation}
\end{enumerate}
We will set $\kappa_0 =0$ for $\tfrac{y}{t}<2\varkappa$ for notational convenience later on.
\item 
$0 < \frac{y}{t}< 2\varkappa$, i.e.\ $0<c<2$.
In this case 
\begin{enumerate}[$\triangleright$]
\item
$k_0\in\D{R}$, $k_1\not\in\D{R}$,
\item
$\Re(\Phi)>0$ for $\Im(k)>0$ near $(-k_0,k_0)\cup \ii(0,\tfrac{1}{2})$
and 
\item
$\Re(\Phi)<0$ for $\Im(k)>0$ near $(-\infty,-k_0) \cup (k_0,\infty)$.
\end{enumerate}
\item 
$-\frac{\varkappa}{4} < \frac{y}{t}< 0$, i.e.\ $-1/4<c<0$.
In this case 
\begin{enumerate}[$\triangleright$]
\item
$k_0,k_1\in\D{R}$,
\item 
$\Re(\Phi)>0$ for $\Im(k)>0$ near $(-\infty,-k_1) \cup(-k_0,k_0) \cup (k_1,\infty) \cup \ii(0,\tfrac{1}{2})$
and 
\item
$\Re(\Phi)<0$ for $\Im(k)>0$ near $(-k_1,-k_0) \cup (k_0,k_1)$.
\end{enumerate}
\item 
$\frac{y}{t}<-\frac{\varkappa}{4}$, i.e.\ $c<-1/4$.
In this case 
\begin{enumerate}[$\triangleright$]
\item
$k_0,k_1 \not\in\D{R}$ and 
\item
$\Re(\Phi)>0$ for $\Im(k)>0$ near $\D{R} \cup \ii(0,\tfrac{1}{2})$.
\end{enumerate}
\end{enumerate}
The situation is depicted in Figure~\ref{fig:signRePhi}.

\begin{figure}[ht]
\begin{picture}(2.2,4.2)
\put(0.3,3.8){$c<-1/4$}
\put(1.2,0.7){$+$}
\put(1.2,2.5){$-$}
\put(0.6,1.4){$-$}
\put(0.6,1.8){$+$}
\put(0,1.7){\vector(1,0){2}}
\put(1,0){\vector(0,1){3.6}}
\closecurve(1.856,2.534, 1.728,2.772, 1.529,2.964, 1.278,3.089, 1.,3.132,
0.722,3.089, 0.471,2.964, 0.272,2.772, 0.144,2.534, 0.1,2.283,
0.144,2.073, 0.272,2.001, 0.471,2.077, 0.722,2.166, 1.,2.2,
1.278,2.166, 1.529,2.077, 1.728,2.001, 1.856,2.073, 1.9,2.283)
\closecurve(1.9,1.117, 1.856,0.866, 1.728,0.628, 1.529,0.436, 1.278,0.311,
1.,0.268, 0.722,0.311, 0.471,0.436, 0.272,0.628, 0.144,0.866,
0.1,1.117, 0.144,1.327, 0.272,1.399, 0.471,1.323, 0.722,1.234,
1.,1.2, 1.278,1.234, 1.529,1.323, 1.728,1.399, 1.856,1.327)
\end{picture}
\quad
\begin{picture}(3,4.2)
\put(0.2,3.8){$-1/4<c< 0$}
\put(1.6,0.5){$+$}
\put(1.6,2.5){$-$}
\put(1,1.4){$-$}
\put(1,1.8){$+$}
\put(0,1.7){\vector(1,0){2.9}}
\put(1.4,0){\vector(0,1){3.6}}
\put(2.06,1.7){\circle*{0.06}}
\put(2.06,1.8){\scriptsize $k_0$}
\put(0.74,1.7){\circle*{0.06}}
\put(0.25,1.8){\scriptsize $-k_0$}
\put(2.65,1.7){\circle*{0.06}}
\put(2.7,1.4){\scriptsize $k_1$}
\put(0.15,1.7){\circle*{0.06}}
\put(-0.35,1.4){\scriptsize $-k_1$}
\closecurve(2.645,1.948, 2.584,2.406, 2.407,2.807, 2.132,3.106, 1.785,3.294,
1.4,3.358, 1.015,3.294, 0.668,3.106, 0.393,2.807, 0.216,2.406,
0.155,1.948, 0.155,1.452, 0.216,0.994, 0.393,0.593, 0.668,0.294,
1.015,0.106, 1.4,0.042, 1.785,0.106, 2.132,0.294, 2.407,0.593,
2.584,0.994, 2.645,1.452)
\closecurve(2.064,1.707, 2.032,1.85, 1.937,1.988, 1.79,2.1, 1.605,2.174, 1.4,2.2,
1.195,2.174, 1.01,2.1, 0.863,1.988, 0.768,1.85, 0.736,1.707,
0.736,1.693, 0.768,1.55, 0.863,1.412, 1.01,1.3, 1.195,1.226, 1.4,1.2,
1.605,1.226, 1.79,1.3, 1.937,1.412, 2.032,1.55, 2.064,1.693)
\end{picture}
\quad
\begin{picture}(2.2,4.2)
\put(0.3,3.8){$0<c<2$}
\put(0.6,0.7){$+$}
\put(0.6,2.5){$-$}
\put(1.1,1.4){$-$}
\put(1.1,1.8){$+$}
\put(0,1.7){\vector(1,0){2}}
\put(1,0){\vector(0,1){3.6}}
\put(1.45,1.7){\circle*{0.06}}
\put(1.5,1.8){\scriptsize $k_0$}
\put(0.55,1.7){\circle*{0.06}}
\put(0.,1.8){\scriptsize $-k_0$}
\closecurve(1.456,1.7, 1.434,1.855, 1.369,1.994, 1.268,2.105, 1.141,2.176,
1.,2.2, 0.859,2.176, 0.732,2.105, 0.631,1.994, 0.566,1.855,
0.544,1.7, 0.566,1.545, 0.631,1.406, 0.732,1.295, 0.859,1.224,
1.,1.2, 1.141,1.224, 1.268,1.295, 1.369,1.406, 1.434,1.545)
\end{picture}
\quad
\begin{picture}(2.2,4.2)
\put(0.5,3.8){$2<c$}
\put(0.6,0.7){$+$}
\put(0.6,2.5){$-$}
\put(0,1.7){\vector(1,0){2}}
\put(1,0){\vector(0,1){3.6}}
\put(1,1.8){\circle*{0.06}}
\put(1.2,1.8){\scriptsize $\kappa_0$}
\put(1,1.6){\circle*{0.06}}
\put(1.1,1.5){\scriptsize $-\kappa_0$}
\closecurve(1.119,2.037, 1.113,1.985, 1.096,1.927, 1.07,1.872, 1.037,1.828,
1.,1.809, 0.963,1.828, 0.93,1.872, 0.904,1.927, 0.887,1.985,
0.881,2.037, 0.881,2.047, 0.887,2.094, 0.904,2.138, 0.93,2.172,
0.963,2.193, 1.,2.2, 1.037,2.193, 1.07,2.172, 1.096,2.138,
1.113,2.094, 1.119,2.047)
\closecurve(1.119,1.363, 1.113,1.415, 1.096,1.473, 1.07,1.528, 1.037,1.572,
1.,1.591, 0.963,1.572, 0.93,1.528, 0.904,1.473, 0.887,1.415,
0.881,1.363, 0.881,1.353, 0.887,1.306, 0.904,1.262, 0.93,1.228,
0.963,1.207, 1.,1.2, 1.037,1.207, 1.07,1.228, 1.096,1.262,
1.113,1.306, 1.119,1.353)
\end{picture}
\caption{Sign of $\Re(\Phi(k))$ for different values of $c=\frac{y}{\varkappa\ t}$}\label{fig:signRePhi}
\end{figure}
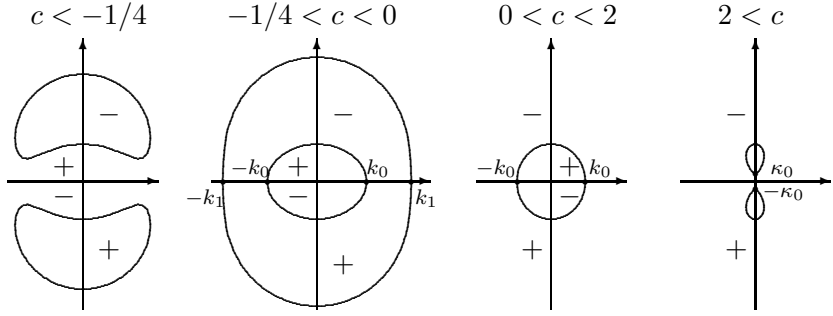

Accordingly we will introduce
\begin{equation}
\Sigma(c) = 
\begin{cases}
\D{R}, & c< - \frac{1}{4},\\
(-\infty,-k_1)\cup(-k_0,k_0)\cup(k_1,\infty), & - \frac{1}{4} <  c < 0,\\
(-k_0,k_0), & 0 < c < 2,\\
\varnothing, & 2 < c.
\end{cases}
\end{equation}
As mentioned above we will need the following factorizations of the jump condition \eqref{eq:jumpcond}:
\begin{equation}
v(k)=b_-(k)^{-1}b_+(k),
\end{equation}
where
\begin{equation}
b_-(k)=\begin{pmatrix} 1 & \overline{R(k)}\ee^{-t\Phi(k)} \\ 0 & 1 \end{pmatrix}, \qquad
b_+(k)=\begin{pmatrix} 1 & 0 \\ R(k)\ee^{t\Phi(k)} & 1\end{pmatrix}.
\end{equation}
for $k\in\D{R}\setminus\Sigma(c)$ and
\begin{equation}
v(k)=B_-(k)^{-1}
\begin{pmatrix} 1-\left\vert R(k) \right\vert^2 & 0 \\
	0 & \frac{1}{1-\left\vert R(k) \right\vert^2}\end{pmatrix}
B_+(k),
\end{equation}
where
\begin{equation}
B_-(k)=\begin{pmatrix} 1 & 0\\
	-\frac{R(k)\ee^{t\Phi(k)}}{1-\left\vert R(k) \right\vert^2} & 1 \end{pmatrix}, \qquad
B_+(k)=\begin{pmatrix} 1 & -\frac{\overline{R(k)}\ee^{-t\Phi(k)}}{1-\left\vert R(k) \right\vert^2}\\
	0 & 1 \end{pmatrix}.
\end{equation}
for $k\in\Sigma(c)$.

To get rid of the diagonal part in the factorization corresponding to
$k\in\Sigma(c)$ and to conjugate the jumps near the eigenvalues we need
the partial transmission coefficient with respect to $c$ defined by
\begin{equation}\label{def:Tkk0}
T(k,c)=
\begin{cases}
\prod\limits_{\kappa_0<\kappa_j<1/2}\frac{k+\ii\kappa_j}{k-\ii\kappa_j},& c>2, \\
\prod\limits_{j=1}^{N} \frac{k+\ii\kappa_j}{k-\ii\kappa_j}
\exp\bigl(\frac{1}{2\pi\ii}\int_{\Sigma(c)}\frac{\log(\abs{T(\zeta)}^2)}{\zeta-k}\,
\dd\zeta\bigr), & c<2,
\end{cases}
\end{equation}
for $k\in\D{C}\setminus\Sigma(c)$. Thus $T(k,c)$
is meromorphic for $k\in\D{C}\setminus\Sigma(c)$. Note that $T(k,c)$ can be computed in terms
of the scattering data since $|T(k)|^2= 1- |R_+(k,t)|^2$. Moreover, we set
\begin{equation}
T_1(c)=\log\bigl(T(\tfrac{\ii}{2},c)\bigr)= 
\begin{cases}
\sum\limits_{\kappa_0<\kappa_j<1/2}\log\frac{1+2\kappa_j}{1-2\kappa_j},& c>2,\\
\sum\limits_{j=1}^{N}\log\frac{1+2\kappa_j}{1-2\kappa_j}
+\frac{1}{\pi}\int_{\Sigma(c)}\frac{\log(|T(\zeta)|^2)}{1+4\zeta^2}\,\dd\zeta, & c<2.
\end{cases}
\end{equation}
Note that combining $T(k) = \ee^{\ii k H_{-1}(w)} T(k,c)$ for  $c< -\frac{1}{4}$ with \eqref{g_ass_06} shows
\begin{equation}
H_{-1}(w) = 2 \sum_{j=1}^{N}\log\frac{1+2\kappa_j}{1-2\kappa_j}
+ \frac{2}{\pi}\int_{\D{R}}\frac{\log(|T(\zeta)|^2)}{1+4\zeta^2}\,\dd\zeta.
\end{equation}

\begin{theorem}\label{thm:part}
The partial transmission coefficient $T(k,c)$ is meromorphic in $\D{C}\setminus\Sigma(c)$ with simple poles
at $\ii\kappa_j$ and simple zeros at $-\ii\kappa_j$ for all j with $\kappa_0 < \kappa_j$, and
satisfies the jump condition
\begin{equation}\label{eq:jumpt}
T_+(k,c)=T_-(k,c)(1-\left\vert R(k) \right\vert^2), \quad \text{ for } k\in\Sigma(c).
\end{equation}
Moreover:
\begin{enumerate}[\rm(i)]
\item
$T(-k,c)=T(k,c)^{-1}$, $k\in\D{C}\setminus\Sigma(c)$.
\item
$T(-k,c)=\overline{T(\overline{k},c)}$, $k\in\D{C}$, in particular $T(k,c)$ is
real for $k\in\ii\D{R}$. 
\item
If $c<2$ the behaviour near $k=0$ is given by $T(k,c)=T(k)(C+o(1))$ with $C\neq 0$ for $\Im(k)\geq 0$.
\end{enumerate}
\end{theorem}

\begin{proof}
That $\ii\kappa_j$ are simple poles and $-\ii\kappa_j$ are simple zeros is obvious from the
Blaschke factors and that $T(k,c)$ has the given jump follows from Plemelj's formulas.
Properties (i), (ii), and (iii) are straightforward to check.
\end{proof}

Now we are ready to perform our conjugation step. Introduce
\[
D(k)= 
\begin{cases}
\begin{pmatrix} 1 & -\frac{k-\ii\kappa_j}{\ii\gamma_j^2 \ee^{t\Phi (\ii\kappa_j)}}\\
\frac{\ii\gamma_j^2 \ee^{t\Phi(\ii\kappa_j)}}{k-\ii\kappa_j} & 0 \end{pmatrix}
D_0(k), 
&\abs{k-\ii\kappa_j}<\varepsilon, \: \kappa_0 < \kappa_j,\\
\begin{pmatrix} 0 & -\frac{\ii\gamma_j^2 \ee^{t\Phi (\ii\kappa_j)}}{k+\ii\kappa_j} \\
\frac{k+\ii\kappa_j}{\ii\gamma_j^2 \ee^{t\Phi(\ii\kappa_j)}} & 1 \end{pmatrix}
D_0(k), 
&\abs{k+\ii\kappa_j}<\varepsilon, \: \kappa_0 < \kappa_j,\\
D_0(k), 
& \text{else},
\end{cases}
\]
where
\[
D_0(k) = \begin{pmatrix} T(k,c)^{-1} & 0 \\ 0 & T(k,c) \end{pmatrix}.
\]
Observe that $D(k)$ respects our symmetry:
\[
D(-k)=\begin{pmatrix}0&1\\1&0\end{pmatrix} D(k)\begin{pmatrix}0&1\\1&0\end{pmatrix}.
\]
Now we conjugate our problem using $D(k)$ and set
\begin{equation}\label{def:mti}
\tilde{m}(k)=m(k) D(k).
\end{equation}
Note that even though $D(k)$ might be singular at $k=0$ (if $c<2$ and $R(0)=-1$),
$\tilde{m}(k)$ is nonsingular since the possible singular behaviour of $T(k,c)^{-1}$ from $D_0(k)$
cancels with $T(k)$ in $m(k)$ by virtue of Theorem~\ref{thm:part} (iii).

Then using Lemma~\ref{lem:conjug} with $\tilde\Sigma=\Sigma(c)$ and Lemma~\ref{lem:twopolesinc} the jump
corresponding to $\kappa_0<\kappa_j$ (if any) is given by
\begin{equation}
\aligned
\tilde{v}(k) &= \begin{pmatrix}1& -\frac{k-\ii\kappa_j}
{\ii\gamma_j^2 \ee^{t\Phi (\ii\kappa_j)}T(k,c)^{-2}}\\ 0 &1\end{pmatrix},
&&\abs{k-\ii\kappa_j}=\varepsilon, \\
\tilde{v}(k) &= \begin{pmatrix}1& 0 \\ -\frac{k+\ii\kappa_j}
{\ii\gamma_j^2 \ee^{t\Phi(\ii\kappa_j)} T(k,c)^2}&1\end{pmatrix},
&&\abs{k+\ii\kappa_j}=\varepsilon,
\endaligned
\end{equation}
and corresponding to $\kappa_0>\kappa_j$ (if any) by
\begin{equation}
\tilde{v}(k)=
\begin{cases}
\begin{pmatrix} 1 & 0 \\ -\frac{\ii\gamma_j^2 \ee^{t\Phi(\ii\kappa_j)} T(k,c)^{-2}}{k-\ii\kappa_j}& 1 \end{pmatrix},
&\abs{k-\ii\kappa_j}=\varepsilon, \\[6mm]
\begin{pmatrix} 1 & -\frac{\ii\gamma_j^2 \ee^{t\Phi(\ii\kappa_j)} T(k,c)^2}{k+\ii\kappa_j} \\
0 & 1 \end{pmatrix},
&\abs{k+\ii\kappa_j}=\varepsilon.
\end{cases}
\end{equation}
In particular, all jumps corresponding to poles, except for possibly one if
$\kappa_j=\kappa_0$, are exponentially close to the identity. In the latter case we will keep the
pole condition for $\kappa_j=\kappa_0$ which now reads
\begin{equation}
\aligned
\Res_{\ii\kappa_j} \tilde{m}(k) &= \lim_{k\to\ii\kappa_j} \tilde{m}(k)
\begin{pmatrix} 0 & 0\\ \ii\gamma_j^2 \ee^{t\Phi(\ii\kappa_j)} T(\ii\kappa_j,c)^{-2}  & 0 \end{pmatrix},\\
\Res_{-\ii\kappa_j} \tilde{m}(k) &= \lim_{k\to -\ii\kappa_j} \tilde{m}(k)
\begin{pmatrix} 0 & -\ii\gamma_j^2 \ee^{t\Phi(\ii\kappa_j)} T(\ii\kappa_j,c)^{-2} \\ 0 & 0 \end{pmatrix}.
\endaligned
\end{equation}
Furthermore, the jump along
$\D{R}$ is given by
\begin{equation}
\tilde{v}(k) = 
\begin{cases}
\tilde{b}_-(k)^{-1} \tilde{b}_+(k),&k\not\in\Sigma(c),\\
\tilde{B}_-(k)^{-1} \tilde{B}_+(k),&k\in\Sigma(c),\\
\end{cases}
\end{equation}
where
\begin{equation} \label{eq:deftib}
\tilde{b}_-(k) = \begin{pmatrix} 1 & \frac{R(-k) \ee^{-t\Phi(k)}}{T(-k,c)^2} \\ 0 & 1 \end{pmatrix}, \quad
\tilde{b}_+(k) = \begin{pmatrix} 1 & 0 \\ \frac{R(k) \ee^{t\Phi(k)}}{T(k,c)^2} & 1 \end{pmatrix},
\end{equation}
and
\begin{align*}
&\tilde{B}_-(k)
=\begin{pmatrix} 1 & 0 \\ -\frac{T_-(k,c)^{-2}}{1-\left\vert R(k) \right\vert^2}R(k)\ee^{t\Phi(k)} & 1\end{pmatrix}
=\begin{pmatrix} 1 & 0 \\ - \frac{T_-(-k,c)}{T_-(k,c)} R(k) \ee^{t\Phi(k)} & 1 \end{pmatrix},\\
&\tilde{B}_+(k)
=\begin{pmatrix} 1 & -\frac{T_+(k,c)^2}{1-\left\vert R(k)\right\vert^2}R(-k)\ee^{-t\Phi(k)} \\ 0 & 1 \end{pmatrix}
=\begin{pmatrix} 1 & - \frac{T_+(k,c)}{T_+(-k,c)} R(-k) \ee^{-t\Phi(k)} \\ 0 & 1 \end{pmatrix}.
\end{align*}
Here we have used
\begin{alignat*}{2}
&R(-k)=\overline{R(k)},&\quad& k\in\D{R},\\
&T_\pm(-k,c)=T_\mp(k,c)^{-1},&&k\in\Sigma(c),
\end{alignat*}
and the jump condition \eqref{eq:jumpt} for the partial transmission coefficient $T(k,c)$ along $\Sigma(c)$ in the last step.
This also shows that the matrix entries are bounded for $k\in\D{R}$ near $k=0$ since $T_\pm(-k,c)=\overline{T_\pm(k,c)}$.

Since we have assumed that $R(k)$ has an analytic continuation to a neighborhood of the real axis, we can now
deform the jump along $\D{R}$ to move the oscillatory terms into regions where they are decaying.

There are four cases to distinguish:

\subsubsection*{Case \emph{(i):} $c>2$}

We set $\Sigma_\pm =\accol{k\in\D{C}\mid\Im(k) = \pm \varepsilon}$ for some small $\varepsilon$ such that $\Sigma_\pm$ lies in the region with $\pm\Re(\Phi(k))<0$ and such that the circles $C_j$, $\bar C_j$ around $\pm\ii\kappa_j$ lie outside the region in between $\Sigma_-$ and $\Sigma_+$.
Then we can split our jump by redefining $\tilde{m}(k)$ according to
\begin{equation}
\widehat{m}(k) = 
\begin{cases} 
\tilde{m}(k)\tilde{b}_+(k)^{-1} ,  &
0<\Im(k)<\varepsilon,\\
\tilde{m}(k) \tilde{b}_-(k)^{-1},  &
-\varepsilon < \Im(k) < 0,\\
\tilde{m}(k), & \text{else}.
\end{cases}
\end{equation}
Thus the jump along the real axis disappears and the jump along $\Sigma_\pm$ is given by
\begin{equation}\label{eq:jumpsolreg}
\widehat{v}(k) = 
\begin{cases} 
\tilde{b}_+(k) , &  k\in\Sigma_+ \\
\tilde{b}_-(k)^{-1} , & k\in\Sigma_-.
\end{cases}
\end{equation}
All other jumps are unchanged. By construction the jump along $\Sigma_{\pm}$ is exponentially close to the identity as $t\to\infty$.

\subsubsection*{Cases \emph{(ii)} and \emph{(iii)}: $0<c<2$, respectively $-1/4<c<0$}
We set $\Sigma_\pm=\Sigma_\pm^1 \cup\Sigma_\pm^2$ according to Figure~\ref{figure:simreg1} respectively
Figure~\ref{figure:simreg2} again such that the circles around $\pm\ii\kappa_j$ lie outside the region in between
$\Sigma_-$ and $\Sigma_+$. 
Again note that $\Sigma_\pm^1$ respectively $\Sigma_\pm^2$ lie in the region with
$\pm\Re(\Phi(k))<0$.

\begin{figure}[hb]
\centering
\begin{picture}(11,5)

\put(0,2.5){\line(1,0){11}}
\put(2.4,2.5){\vector(1,0){0.4}}
\put(8,2.5){\vector(1,0){0.4}}

\put(10.8,2.2){$\D{R}$}

\put(0.5,2){\vector(1,0){0.4}}
\put(2.4,2){\vector(1,0){0.4}}
\put(5.3,2){\vector(1,0){0.4}}
\put(8.3,2){\vector(1,0){0.4}}
\put(10.3,2){\vector(1,0){0.4}}

\put(0.5,3){\vector(1,0){0.4}}
\put(2.4,3){\vector(1,0){0.4}}
\put(5.3,3){\vector(1,0){0.4}}
\put(8.3,3){\vector(1,0){0.4}}
\put(10.3,3){\vector(1,0){0.4}}

\put(0.5,1.6){$\Sigma_-^2$}
\put(2.4,1.6){$\Sigma_-^1$}
\put(5.3,1.6){$\Sigma_-^2$}
\put(8.3,1.6){$\Sigma_-^1$}
\put(10.3,1.6){$\Sigma_-^2$}

\put(0.5,3.2){$\Sigma_+^2$}
\put(2.4,3.2){$\Sigma_+^1$}
\put(5.3,3.2){$\Sigma_+^2$}
\put(8.3,3.2){$\Sigma_+^1$}
\put(10.3,3.2){$\Sigma_+^2$}

\put(0.6,1.0){$\scriptstyle\Re\Phi<0$}
\put(0.6,3.8){$\scriptstyle\Re\Phi>0$}
\put(2.3,1.0){$\scriptstyle\Re\Phi>0$}
\put(2.3,3.8){$\scriptstyle\Re\Phi<0$}
\put(5.0,3.6){$\scriptstyle\Re\Phi>0$}
\put(5.0,1.2){$\scriptstyle\Re\Phi<0$}
\put(7.5,1.0){$\scriptstyle\Re\Phi>0$}
\put(7.5,3.8){$\scriptstyle\Re\Phi<0$}
\put(9.7,1.0){$\scriptstyle\Re\Phi<0$}
\put(9.7,3.8){$\scriptstyle\Re\Phi>0$}

\put(3.7,2.6){$\scriptstyle -k_0$}
\put(7.1,2.6){$\scriptstyle k_0$}
\put(1.1,2.6){$\scriptstyle -k_1$}
\put(9.5,2.6){$\scriptstyle k_1$}

\put(0,3){\line(1,0){1.3}}
\put(0,2){\line(1,0){1.3}}
\curve(1.3,2., 1.5,2.1, 1.7,2.4, 1.75,2.5, 1.8,2.6, 2.,2.9, 2.2,3.)
\curve(1.3,3., 1.5,2.9, 1.7,2.6, 1.75,2.5, 1.8,2.4, 2.,2.1, 2.2,2.)
\put(2.2,3){\line(1,0){0.86}}
\put(2.2,2){\line(1,0){0.86}}
\curve(3.06,2., 3.26,2.1, 3.46,2.4, 3.51,2.5, 3.56,2.6, 3.76,2.9, 3.96,3.)
\curve(3.06,3., 3.26,2.9, 3.46,2.6, 3.51,2.5, 3.56,2.4, 3.76,2.1, 3.96,2.)
\put(3.96,3){\line(1,0){3.08}}
\put(3.96,2){\line(1,0){3.08}}
\curve(7.04,2., 7.24,2.1, 7.44,2.4, 7.49,2.5, 7.54,2.6, 7.74,2.9, 7.94,3.)
\curve(7.04,3., 7.24,2.9, 7.44,2.6, 7.49,2.5, 7.54,2.4, 7.74,2.1, 7.94,2.)
\put(7.94,3){\line(1,0){0.88}}
\put(7.94,2){\line(1,0){0.88}}
\curve(8.82,2., 9.02,2.1, 9.22,2.4, 9.27,2.5, 9.32,2.6, 9.52,2.9, 9.72,3.)
\curve(8.82,3., 9.02,2.9, 9.22,2.6, 9.27,2.5, 9.32,2.4, 9.52,2.1, 9.72,2.)
\put(9.72,3){\line(1,0){1.28}}
\put(9.72,2){\line(1,0){1.28}}

\curvedashes{0.05,0.05}

\curve(1.948,4.618, 1.765,3.245, 1.765,1.755, 1.948,0.382)
\curve(9.052,4.618, 9.235,3.245, 9.235,1.755, 9.052,0.382)

\closecurve(7.492,2.52, 7.395,2.95, 7.112,3.363, 6.671,3.701, 6.116,3.923,
5.5,4., 4.884,3.923, 4.329,3.701, 3.888,3.363, 3.605,2.95,
3.508,2.52, 3.508,2.48, 3.605,2.05, 3.888,1.637, 4.329,1.299,
4.884,1.077, 5.5,1., 6.116,1.077, 6.671,1.299, 7.112,1.637,
7.395,2.05, 7.492,2.48)

\end{picture}
\caption{Deformed contour for $-1/4<c<0$}
\label{figure:simreg1}
\end{figure}
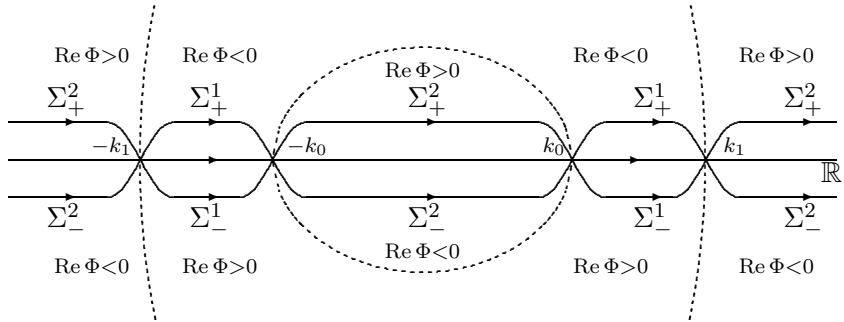

\begin{figure}[ht]
\centering
\begin{picture}(7,4.5)
\put(0,2.5){\line(1,0){7.0}}
\put(1.5,2.5){\vector(1,0){0.4}}
\put(5,2.5){\vector(1,0){0.4}}

\put(6.8,2.2){$\D{R}$}

\put(0,2){\line(1,0){2.0}}
\put(2.9,3){\line(1,0){1.2}}
\put(5,2){\line(1,0){2.0}}
\put(1.1,2){\vector(1,0){0.4}}
\put(5.5,2){\vector(1,0){0.4}}
\put(3.3,3){\vector(1,0){0.4}}

\put(1.3,1.6){$\Sigma_-^1$}
\put(5.7,1.6){$\Sigma_-^1$}
\put(3.5,3.2){$\Sigma_+^2$}

\curve(2.,2., 2.2,2.1, 2.4,2.4, 2.45,2.5, 2.5,2.6, 2.7,2.9, 2.9,3.)

\curve(4.1,3., 4.3,2.9, 4.5,2.6, 4.55,2.5, 4.6,2.4, 4.8,2.1, 5.,2.)

\put(0,3){\line(1,0){2.0}}
\put(2.9,2){\line(1,0){1.2}}
\put(5,3){\line(1,0){2.0}}
\put(1.1,3){\vector(1,0){0.4}}
\put(5.5,3){\vector(1,0){0.4}}
\put(3.3,2){\vector(1,0){0.4}}

\curve(2.,3., 2.2,2.9, 2.4,2.6, 2.45,2.5, 2.5,2.4, 2.7,2.1, 2.9,2.)

\curve(4.1,2., 4.3,2.1, 4.5,2.4, 4.55,2.5, 4.6,2.6, 4.8,2.9, 5.,3.)

\put(1.3,3.2){$\Sigma_+^1$}
\put(5.7,3,2){$\Sigma_+^1$}
\put(3.5,1.6){$\Sigma_-^2$}

\put(0.3,1.0){$\scriptstyle\Re\Phi>0$}
\put(0.3,3.8){$\scriptstyle\Re\Phi<0$}
\put(5.9,1.0){$\scriptstyle\Re\Phi>0$}
\put(5.9,3.8){$\scriptstyle\Re\Phi<0$}
\put(3,3.6){$\scriptstyle\Re\Phi>0$}
\put(3,1.2){$\scriptstyle\Re\Phi<0$}

\put(2.6,2.6){$\scriptstyle -k_0$}
\put(4.1,2.6){$\scriptstyle k_0$}

\curvedashes{0.05,0.05}

\closecurve(4.55,2.5, 4.498,2.979, 4.349,3.402, 4.117,3.727, 3.824,3.931, 3.5,4.,
3.176,3.931, 2.883,3.727, 2.651,3.402, 2.502,2.979, 2.45,2.5,
2.502,2.021, 2.651,1.598, 2.883,1.273, 3.176,1.069, 3.5,1.,
3.824,1.069, 4.117,1.273, 4.349,1.598, 4.498,2.021)

\end{picture}
\caption{Deformed contour for $0<c<2$}
\label{figure:simreg2}
\end{figure}
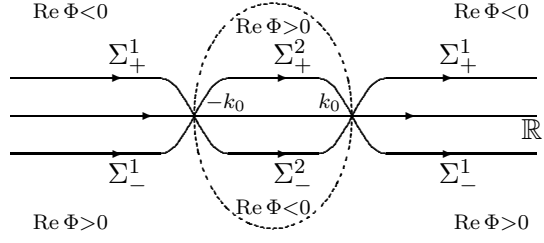

Then we can split our jump by redefining $\tilde{m}(k)$ according to
\begin{equation}
\widehat{m}(k)= \begin{cases} \tilde{m}(k)\tilde{b}_+(k)^{-1} , &  k\text{ between } \D{R} \text{ and } \Sigma_+^1, \\
\tilde{m}(k)\tilde{b}_-(k)^{-1} , & k \text{ between } \D{R} \text{ and } \Sigma_-^1,\\
\tilde{m}(k)\tilde{B}_+(k)^{-1} , & k \text{ between } \D{R} \text{ and } \Sigma_+^2,\\
\tilde{m}(k)\tilde{B}_-(k)^{-1} , & k \text{ between } \D{R} \text{ and } \Sigma_-^2,\\
\tilde{m}(k) , & \text{else}.
\end{cases}
\end{equation}

One checks that the jump along $\D{R}$ disappears and the jump along $\Sigma_{\pm}$ is given by
\begin{equation}
\widehat{v}(k)=\begin{cases} \tilde{b}_+(k) , & k\in\Sigma_+^1,\\
\tilde{b}_-(k)^{-1} , & k\in\Sigma_-^1,\\
\tilde{B}_+(k) , & k\in\Sigma_+^2,\\
\tilde{B}_-(k)^{-1} , &  k\in\Sigma_-^2. 
\end{cases}
\end{equation}

All other jumps are unchanged. Again the resulting Riemann--Hilbert problem still satisfies our symmetry condition \eqref{eq:symcond} and the jump along $\Sigma_{\pm}\setminus\{\pm k_0, \pm k_1\}$ is exponentially decreasing as $t\to\infty$.

\subsubsection*{Case \emph{(iv)}: $c<-1/4$}

We set $\Sigma_\pm =\accol{k\in\D{C}\mid\Im(k)=\pm\varepsilon}$ for some small $\varepsilon$ such that
$\Sigma_\pm$ lies in the region with $\pm \Re(\Phi(k)) > 0$ and such that the circles around $\pm\ii\kappa_j$
lie outside the region in between $\Sigma_-$ and $\Sigma_+$.
Then we can split our jump by redefining $\tilde{m}(k)$ according to
\begin{equation}
\widehat{m}(k) = 
\begin{cases} 
\tilde{m}(k)\tilde{B}_+(k)^{-1} ,  &0<\Im(k)<\varepsilon,\\
\tilde{m}(k) \tilde{B}_-(k)^{-1},  &-\varepsilon < \Im(k) < 0,\\
\tilde{m}(k), & \text{else}.
\end{cases}
\end{equation}
Thus the jump along the real axis disappears and the jump along $\Sigma_\pm$ is given by
\begin{equation}
\widehat{v}(k) = 
\begin{cases} 
\tilde{B}_+(k) , &  k\in\Sigma_+ \\
\tilde{B}_-(k)^{-1} , & k\in\Sigma_-.
\end{cases}
\end{equation}
All other jumps are unchanged. By construction the jump along $\Sigma_{\pm}$ is exponentially close to the identity as $t\to\infty$.

Note that in all cases the resulting Riemann--Hilbert problem still satisfies our symmetry condition \eqref{eq:symcond}, since we have
\begin{equation}
\tilde{b}_\pm(-k)= \begin{pmatrix}0&1\\1&0\end{pmatrix} \tilde{b}_\mp(k)\begin{pmatrix}0&1\\1&0\end{pmatrix}, \quad \tilde{B}_\pm(-k)= \begin{pmatrix}0&1\\1&0\end{pmatrix} \tilde{B}_\mp(k)\begin{pmatrix}0&1\\1&0\end{pmatrix}.
\end{equation}

In Cases~(i) and (iv) we can immediately apply Theorem~\ref{thm:remcontour} to $\widehat{m}$ as follows:

\begin{proof}[Proof of Theorem~\emph{\ref{thm:main}--\ref{thm:main.momentum} (iv)}]

Since $\widehat{v}(k)= \id + \ord(t^{-l})$ for any $l$ as $t\to\infty$, the same is true for $\widehat{m}(k)= \begin{pmatrix} 1 & 1\end{pmatrix} + \ord(t^{-l})$
by Theorem~\ref{thm:remcontour} (for the case $\gamma=0$). Hence
\begin{align}\notag
m_1(k) m_2(k) &= \widehat{m}_1(k) \widehat{m}_2(k) = 1 + \ord(t^{-l}), \\
\frac{m_1(\frac{\ii}{2})}{m_2(\frac{\ii}{2})} &= \ee^{2 T_1(c)}\frac{\widehat{m}_1(\frac{\ii}{2})}{\widehat{m}_2(\frac{\ii}{2})} = \ee^{2 T_1(c)} + \ord(t^{-l}),
\end{align}
for $k$ near $\frac{\ii}{2}$ and the claim follows from Lemma~\ref{lem:asymp} in case the reflection coefficient has an
analytic extensions.

Otherwise one has to split the reflection coefficient into an analytic part plus a
small remainder. One can literally follow the argument of \cite[Lemma~6.1]{gt}.
\end{proof}

\begin{proof}[Proof of Theorem~\emph{\ref{thm:main}--\ref{thm:main.momentum} (i)}]

If $|\frac{x}{t} - c_j|>\varepsilon$ for all $j$ we can choose $\gamma=0$ in Theorem~\ref{thm:remcontour}. Hence as in the proof of (iv),
\begin{align}\notag
m_1(k) m_2(k) &= \widehat{m}_1(k) \widehat{m}_2(k) = 1 + \ord(t^{-l}), \\
\frac{m_1(\frac{\ii}{2})}{m_2(\frac{\ii}{2})} &= \ee^{2 T_1(c)} \frac{\widehat{m}_1(\frac{\ii}{2})}{\widehat{m}_2(\frac{\ii}{2})} =
\ee^{2 T_1(c)} + \ord(t^{-l}),
\end{align}
for $k$ near $\frac{\ii}{2}$ and the claim follows as before.

Otherwise, if $|\frac{x}{t} - c_j|<\varepsilon$ for some $j$, we choose $\gamma=\gamma_j(x,t)$. Again we conclude
\begin{align}\notag
m_1(k) m_2(k) &= \widehat{m}_1(k) \widehat{m}_2(k) = f(k)f(-k) + \ord(t^{-l}), \\
\frac{m_1(\frac{\ii}{2})}{m_2(\frac{\ii}{2})} &= \ee^{2 T_1(c)} \frac{\widehat{m}_1(\frac{\ii}{2})}{\widehat{m}_2(\frac{\ii}{2})} =
\ee^{2 T_1(c)} \frac{f(\frac{\ii}{2})}{f(-\frac{\ii}{2})}+ \ord(t^{-l}),
\end{align}
where $f(k)$ is the one-soliton solution from Lemma~\ref{lem:singlesoliton}.
\end{proof}

In the cases (ii) and (iii) the jump will not decay on the small crosses containing the stationary phase
points and we need to continue the investigation of this problem in the next section.

\section{Reduction to a Riemann--Hilbert problem on a small cross}\label{sec:simrhp}

In the previous section we have seen that for $-1/4<c<2$ we can reduce everything to a
Riemann--Hilbert problem for $\widehat{m}(k)$ such that the jumps are exponentially close to the identity
except in small neighborhoods of the stationary phase points $\pm k_0$ and $\pm k_1$. Hence we
need to continue our investigation of this case in this section.

Denote by $\Sigma^c(\pm k_{\ell})$, $\ell=0,1$ the parts of $\Sigma_+\cup \Sigma_-$ inside a small neighborhood of $\pm k_{\ell}$.
We will now show how to solve the two problems on the small crosses $\Sigma^c(k_{\ell})$ respectively
$\Sigma^c(-k_{\ell})$ by reducing them to Theorem~\ref{thm:solcross}. This will lead us to the solution of our original problem by
virtue of Theorem~\ref{thm:decoup}.

Now let us turn to the solution of the problem on
\[
\Sigma^c(k_{\ell})=(\Sigma_+\cup\Sigma_-)\cap\accol{k\mid\abs{k-k_{\ell}}<\varepsilon}
\]
for some small $\varepsilon > 0$. We can also deform our contour slightly such
that $\Sigma^c(k_{\ell})$ consists of two straight lines.
Next, abbreviate
\begin{align}\notag
\Phi_{\ell} &= (-1)^{\ell}\frac{\Phi(k_{\ell})}{\ii} = -(-1)^{\ell}\frac{2 \varkappa k_{\ell}^3}{(1/4+k_{\ell}^2)^2}, \\
\Phi''_{\ell} &= (-1)^{\ell}\frac{\Phi''(k_{\ell})}{\ii} = (-1)^{\ell}\frac{2 \varkappa k_{\ell}(3/4-k_{\ell}^2)}{(1/4+k_{\ell}^2)^3},
\end{align}
where $\Phi_0''>0$ for $-1/4<c<2$ and $\Phi_1''>0$ for $-1/4<c<0$.

As a first step we make a change of coordinates
\begin{equation}\label{eq:zeta}
\zeta_{\ell}=(-1)^{\ell}\sqrt{\Phi''_{\ell}}(k-k_{\ell}) , \qquad
k=k_{\ell}+ \frac{(-1)^{\ell}}{\sqrt{\Phi''_{\ell}}} \zeta_{\ell}
\end{equation}
such that the phase reads $\Phi(k)= (-1)^{\ell}\,\ii\,(\Phi_{\ell}+\frac{1}{2}\zeta^2+\ord(\zeta^3))$.

Next we need the behavior of our jump matrix near $k_{\ell}$, that is, the behavior of $T(k,c)$ near $k_{\ell}$.

\begin{lemma}
We have
\begin{equation}
T(k,c)= 
\begin{cases}
\left(\frac{k-k_0}{k+k_0}\right)^{\ii\nu_0} \left(-\frac{k-k_1}{k+k_1}\right)^{-\ii\nu_1}\tilde{T}(k,c), & -1/4<c<0,\\
\left(\frac{k-k_0}{k+k_0}\right)^{\ii\nu_0} \tilde{T}(k,c), & 0<c<2,
\end{cases}
\end{equation}
where $\nu_{\ell}=-\frac{1}{\pi}\log(|T(k_{\ell})|)>0$ and the branch cut
of the logarithm is chosen along the negative real axis.
Here
\begin{equation}
\tilde{T}(k,c)= \prod_{j=1}^{N}\frac{k+\ii\kappa_j}{k-\ii\kappa_j}
\exp\biggl(-\frac{1}{2\pi\ii}\int_{\Sigma(c)}\log(|k-\zeta|)\,\dd\log(\abs{T(\zeta)}^2)\biggr).
\end{equation}
The function $\tilde{T}(\,\cdot\,,c)$ is H\"older continuous of any exponent less
than $1$ at the stationary phase points $k=k_{\ell}$ and satisfies $\abs{\tilde{T}(k_{\ell},c)}=1$.
\end{lemma}

\begin{proof}
This is a straightforward calculation.
H\"older continuity of any exponent less than $1$ is well-known (cf.~\cite{mu}).
\end{proof}

If $k(\zeta_{\ell})$ is defined as in \eqref{eq:zeta} and $0<\alpha<1$, then there is an $L>0$ such that
\begin{equation}
\left| T(k(\zeta_{\ell}),c)-\zeta_{\ell}^{(-1)^{\ell}\ii\nu_{\ell}}\tilde{T}_{\ell}(c)\ee^{-(-1)^{\ell}\ii\nu_{\ell}\log(2 k_{\ell}\sqrt{\Phi''_{\ell}})}
\right| \leq L\abs{\zeta_{\ell}}^{\alpha},
\end{equation}
where the branch cut of $\zeta_{\ell}^{\ii\nu_{\ell}}$ is chosen along the negative real axis and
\begin{equation}
\tilde{T}_{\ell}(c) = \Bigl((-1)^{\ell}\frac{k-k_{\ell}}{k+k_{\ell}}\Bigr)^{-(-1)^{\ell}\ii\nu_{\ell}} T(k,c) \Big|_{k=k_{\ell}}.
\end{equation}
We also have
\begin{equation}
\left| R(k(\zeta_{\ell}))-R(k_{\ell}) \right| \leq L |\zeta_{\ell}|^{\alpha}.
\end{equation}
Set
\begin{equation}
r_{\ell}=R(k_{\ell})\tilde{T}_{\ell}(c)^{-2}\ee^{(-1)^{\ell}\ii\nu_{\ell} \log(4k_{\ell}^2 \Phi''_{\ell})}
\end{equation}
and note $\nu_{\ell}= -\frac{1}{2\pi}\log(1-|R(k_{\ell})|^2)$ since $|r_{\ell}|=|R(k_{\ell})|$.
Then the assumptions of Theorem~\ref{thm:solcross} are satisfied with $r=r_0$ near for $\Sigma^c(k_0)$.
Similarly, for $\Sigma^c(k_1)$ they are satisfied with $r=\overline{r_1}$ after a conjugation with $(\begin{smallmatrix} 0 & 1\\1 & 0\end{smallmatrix})$.

Therefore we can conclude that the solution on $\Sigma^c(k_{\ell})$ is given by
\begin{align}\notag
M_{\ell}^c(k) & =\id + \frac{1}{\zeta_{\ell}}\frac{\ii}{t^{1/2}}\begin{pmatrix} 0 &  -\beta_{\ell} \\ \overline{\beta_{\ell}} & 0 \end{pmatrix} + \ord(t^{-\alpha})\\
& = \id+ \frac{1}{\sqrt{\Phi''_{\ell}}(k-k_{\ell})}\frac{\ii}{t^{1/2}} \begin{pmatrix} 0 & -\beta_{\ell} \\
\overline{\beta_{\ell}} & 0 \end{pmatrix} + \ord(t^{-\alpha}), \quad 1/2<\alpha<1,
\end{align}
where $\beta_{\ell}$ is given by
\begin{align}\notag
\beta_0 
&= \sqrt{\nu_{\ell}}\ee^{\ii (\pi/4 - \arg(r_0) +\arg(\Gamma(\ii\nu_0)) - t\Phi_0)}t^{-\ii\nu_0}\\
& = \sqrt{\nu_0}\ee^{\ii(\pi/4- \arg(R(k_0)) + 2\arg(\tilde{T}_0(c))-\nu_0\log(4k_0^2\Phi''_0) +
\arg(\Gamma(\ii\nu_0)) -t \Phi_0)}t^{-\ii\nu_0},\\ \notag
\beta_1 & = \sqrt{\nu_1}\ee^{\ii (-\pi/4 - \arg(r_1) -\arg(\Gamma(\ii\nu_1)) + t\Phi_1)}t^{\ii\nu_1}\\
& = \sqrt{\nu_{\ell}}\ee^{\ii(-\pi/4- \arg(R(k_1)) + 2\arg(\tilde{T}_1(c)) + \nu_1\log(4k_1^2\Phi''_1) -
\arg(\Gamma(\ii\nu_1)) + t\Phi_1)}t^{\ii\nu_1}.
\end{align}
We also need the solution $\bar{M}_{\ell}^c(k)$ on $\Sigma^c(-k_{\ell})$. We make the following ansatz, which
is inspired by the symmetry condition for the vector Riemann--Hilbert problem, outside the two small crosses:
\begin{align}\notag
\bar{M}_{\ell}^c(k) &= \begin{pmatrix}0&1\\1&0\end{pmatrix} M_{\ell}^c(-k) \begin{pmatrix}0&1\\1&0\end{pmatrix}\\
 &= \id - \frac{(-1)^{\ell}}{\sqrt{\Phi''_{\ell}}(k+k_{\ell})}\frac{\ii}{t^{1/2}}\begin{pmatrix} 0 & \overline{\beta_{\ell}} \\
-\beta_{\ell} & 0 \end{pmatrix} +\ord(t^{-\alpha}).
\end{align}
Now we are ready to finish the proof of Theorem~\ref{thm:main}-\ref{thm:main.momentum}.

\begin{proof}[Proof of Theorem~\emph{\ref{thm:main}-\ref{thm:main.momentum} (ii)}]
By Theorem~\ref{thm:decoup} we infer
\begin{align}\notag
\widehat{m}(k) =& \begin{pmatrix} 1 & 1 \end{pmatrix}
+ \frac{1}{\sqrt{\Phi''_0 t}}\frac{\ii}{k-k_0} \begin{pmatrix} \overline{\beta_0} & -\beta_0 \end{pmatrix}
- \frac{1}{\sqrt{\Phi''_0 t}}\frac{\ii}{k+k_0} \begin{pmatrix} -\beta_0 & \overline{\beta_0} \end{pmatrix}\\
& + \ord(t^{-\alpha})
\end{align}
and thus
\begin{align}\notag
\widehat{m}(\tfrac{\ii}{2}) & = \begin{pmatrix} 1 & 1 \end{pmatrix} + \frac{4}{\sqrt{\Phi''_0 t}} \Re
\begin{pmatrix} \frac{\beta_0}{1-2\ii k_0} & -\frac{\beta_0}{1+2\ii k_0} \end{pmatrix}+ \ord(t^{-\alpha}),\\
\widehat{m}'(\tfrac{\ii}{2}) & = \frac{8\ii}{\sqrt{\Phi''_0 t}} \Re
\begin{pmatrix} \frac{\beta_0}{(1-2\ii k_0)^2} & -\frac{\beta_0}{(1+2\ii k_0)^2} \end{pmatrix} + \ord(t^{-\alpha}).
\end{align}
Hence
\begin{align}\notag
m_1(k) m_2(k) = & \widehat{m}_1(k) \widehat{m}_2(k) = 1 - \frac{1}{\sqrt{\Phi''_0 t}} \frac{4k_0}{1/4+k_0^2} \Im(\beta_0)+ \ord(t^{-\alpha})\\
&- \left(\frac{1}{\sqrt{\Phi''_0 t}} \frac{ 4k_0\ii}{(1/4+k_0^2)^2} \Im(\beta_0) + \ord(t^{-\alpha})\right)(k-\tfrac{\ii}{2}) + \dots
\end{align}
for $k$ near $\frac{\ii}{2}$ and Lemma~\ref{lem:asymp} implies
\begin{align}\notag
w(x,t) &= \varkappa - \frac{\varkappa}{\sqrt{\Phi''_0 t}} \frac{8 k_0}{1/4+k_0^2} \Im(\beta_0) + \ord(t^{-\alpha})\\
u(x,t) &= -\frac{\varkappa}{\sqrt{\Phi''_0 t}} \frac{2 k_0}{(1/4+k_0^2)^2} \Im(\beta_0) + \ord(t^{-\alpha})
\end{align}
Similarly,
\begin{align}\notag
\frac{m_1(\frac{\ii}{2})}{m_2(\frac{\ii}{2})} &= \ee^{2 T_1(c)}\frac{\widehat{m}_1(\frac{\ii}{2})}{\widehat{m}_2(\frac{\ii}{2})} =
\ee^{2 T_1(c)} \left(1 + \frac{1}{\sqrt{\Phi''_0 t}} \frac{2}{1/4+k_0^2} \Re(\beta_0) \right) + \ord(t^{-\alpha}),
\end{align}
implies
\begin{equation}
x-y = 2T_1(c) + \frac{1}{\sqrt{\Phi''_0 t}} \frac{2}{1/4+k_0^2} \Re(\beta_0) + \ord(t^{-\alpha}).
\end{equation}
Thus the claim follows in case the reflection coefficient has an analytic extensions.
Otherwise one has to split the reflection coefficient into an analytic part plus a
small remainder using \cite[Lem.~6.2 and 6.3]{gt}. Again one can literally follow the argument given there.
\end{proof}

\begin{proof}[Proof of Theorem~\emph{\ref{thm:main}--\ref{thm:main.momentum} (iii)}]
This follows as in the previous case using
\begin{align}\notag
\widehat{m}(k) =& \begin{pmatrix} 1 & 1 \end{pmatrix}
+ \frac{1}{\sqrt{\Phi''_0 t}}\frac{\ii}{k-k_0} \begin{pmatrix} \overline{\beta_0} & -\beta_0 \end{pmatrix}
- \frac{1}{\sqrt{\Phi''_0 t}}\frac{\ii}{k+k_0} \begin{pmatrix} -\beta_0 & \overline{\beta_0} \end{pmatrix}\\ \notag
& + \frac{1}{\sqrt{\Phi''_1 t}}\frac{\ii}{k-k_1} \begin{pmatrix} \overline{\beta_1} & -\beta_1 \end{pmatrix}
- \frac{1}{\sqrt{\Phi''_1 t}}\frac{\ii}{k+k_1} \begin{pmatrix} -\beta_1 & \overline{\beta_1} \end{pmatrix}\\
& + \ord(t^{-\alpha}).
\end{align}
and hence
\begin{align}\notag
w(x,t) &= \varkappa - \frac{\varkappa}{\sqrt{\Phi''_0 t}} \frac{8k_0}{1/4+k_0^2} \Im(\beta_0)
- \frac{\varkappa}{\sqrt{\Phi''_1 t}} \frac{8k_1}{1/4+k_1^2} \Im(\beta_1) + \ord(t^{-\alpha})\\
u(x,t) &= -\frac{\varkappa}{\sqrt{\Phi''_0 t}} \frac{2k_0}{(1/4+k_0^2)^2} \Im(\beta_0)
-\frac{\varkappa}{\sqrt{\Phi''_1 t}} \frac{2k_1}{(1/4+k_1^2)^2} \Im(\beta_1) + \ord(t^{-\alpha})
\end{align}
respectively
\begin{equation}
x-y = 2T_1(c) + \frac{1}{\sqrt{\Phi''_0 t}} \frac{2}{1/4+k_0^2} \Re(\beta_0)
+\frac{1}{\sqrt{\Phi''_1 t}} \frac{2}{1/4+k_1^2} \Re(\beta_1) + \ord(t^{-\alpha}).
\end{equation}
\end{proof}

\appendix
\section{Some results for Riemann--Hilbert problems}       \label{sec:sieq}

In this section we state the results required to prove or main theorems. We will assume that $\Sigma$ is a nice contour, say a finite number of smooth oriented finite curves in $\D{C}$, which intersect at most finitely many times with all intersections being transversal. Moreover, suppose the distance between $\Sigma$ and $\{ \ii y \mid y\geq y_0\}$ is positive for some $y_0>0$.

\subsection{Riemann--Hilbert problem for the soliton region}
The first result is needed in the soliton region. 
Consider the Riemann--Hilbert problem of finding a function $m(k)$ satisfying
\begin{enumerate}[(i)]
\item
$m(k)$ is sectionally meromorphic with simple poles at $\pm\ii\kappa\not\in\Sigma$,
\item
the jump condition
\[
m_+(k) = m_-(k) v(k), \qquad k\in \Sigma,
\]
\item
the pole condition
\[
\Res_{\ii\kappa} m(k) = \lim_{k\to\ii\kappa} m(k)
\begin{pmatrix} 0 & 0\\ \ii\gamma^2  & 0 \end{pmatrix},
\]
\item
the symmetry condition
\[
m(-k) = m(k) \begin{pmatrix}0&1\\1&0\end{pmatrix},
\]
\item
the normalization condition
\[
\lim_{k\to\infty} m(k) = \begin{pmatrix} 1 & 1\end{pmatrix}.
\]
\end{enumerate}
Clearly the symmetry condition implies that our jump data $(\Sigma, v)$ should be symmetric as well, that is,
\begin{equation}
v(-k) = \begin{pmatrix}0&1\\1&0\end{pmatrix} v(k)^{-1} \begin{pmatrix}0&1\\1&0\end{pmatrix},\quad k\in\Sigma.
\end{equation}
and $\Sigma$ is invariant under $k\mapsto -k$ and oriented such that under the
mapping $k\mapsto -k$ sequences converging from the positive sided to $\Sigma$
are mapped to sequences converging to the negative side.
Moreover (ii) and (iii) imply 
\[
\Res_{-\ii\kappa} m(k) = \lim_{k\to -\ii\kappa} m(k)
\begin{pmatrix} 0 & -\ii\gamma^2  \\ 0 & 0 \end{pmatrix}.
\]
Finally we will assume $\norm{v-\id}_2<\infty$ and $\norm{v-\id}_\infty<\infty$.

We are interested in comparing this Riemann--Hilbert problem with the one-soliton problem (where $v\equiv 0$) in the case where $\norm{v-\id}_\infty$ and $\norm{v-\id}_2$ are small. For such a situation we have the following result:

\begin{theorem}[\cite{gt,kt}]\label{thm:remcontour}
Assume $v=v(t)$ satisfies
\begin{equation}
\begin{split}
&\norm{v(t)-\id}_{\infty} \leq \rho(t),\\
&\norm{v(t)-\id}_2\leq\rho(t)
\end{split}
\end{equation}
for some function $\rho(t) \to 0$ as $t\to\infty$. Then the above Riemann--Hilbert problem
has a unique solution for sufficiently large $t$ and the solution differs from the one-soliton
solution by $\ord(\rho(t))$ uniformly in $k$ away from $\Sigma \cup \{\pm\ii\kappa\}$.
\end{theorem}

\subsection{Riemann--Hilbert problem for oscillatory regions}
For the oscillatory regions we will need the following result which allows us to reduce everything to a model problem whose solution will be given below.

\begin{theorem}[\cite{gt,kt}]\label{thm:decoup}
Consider the vector Riemann--Hilbert problem 
\begin{equation}
\begin{split}
& m_+(k)=m_-(k)v(k), \qquad k\in\Sigma, \\
&\lim_{k\to\infty}m(k)=\begin{pmatrix}1&1\end{pmatrix},
\end{split}
\end{equation}
with $\det(v)\neq 0$ and let $0<\alpha<\beta \leq 2\alpha$, $\rho(t)\to\infty$ be given.

Suppose that for every sufficiently small $\varepsilon >0$ both the $L^2$ and the $L^{\infty}$
norms of $v$ are $\ord(t^{-\beta})$ away from some $\varepsilon$ neighborhoods of some points 
$k_i\in\Sigma$, $1\leq i\leq n$. Moreover, suppose that the solution of the matrix problem with jump 
$v(k)$ restricted to the $\varepsilon$ neighborhood of $k_i$ has a solution which satisfies 
\begin{equation}
M_i(k)=\id +\frac{1}{\rho(t)^{\alpha}}\frac{M_i}{k-k_i}+\ord(\rho(t)^{-\beta}), \qquad 
\abs{k-k_i}>\varepsilon.
\end{equation}
Then the solution $m(k)$ is given by 
\begin{equation}
m(k)=\begin{pmatrix} 1 & 1 \end{pmatrix} + 
\frac{1}{\rho(t)^{\alpha}}\begin{pmatrix} 1 & 1 \end{pmatrix} 
\sum_{i=1}^n \frac{M_i}{k-k_i} +
\ord\bigl(\rho(t)^{-\beta}\bigr),
\end{equation}
where the error term depends on the distance of $k$ to $\Sigma$.
\end{theorem}

\subsection*{Model problem on a cross}
Introduce the cross $\Sigma = \Sigma_1 \cup\dots\cup \Sigma_4$ (see Figure~\ref{fig:contourcross}) by
\begin{equation}
\begin{split}
\Sigma_1 & = \accol{u\ee^{-\ii\pi/4}\mid u\geq 0},\qquad
\Sigma_2  = \accol{u\ee^{\ii\pi/4}\mid u\geq 0}, \\
\Sigma_3 & = \accol{u\ee^{3\ii\pi/4}\mid u\geq 0},\qquad\,
\Sigma_4 = \accol{u\ee^{-3\ii\pi/4}\mid u\geq 0}.
\end{split}
\end{equation}

\begin{figure}[ht]
\begin{picture}(7,5.2)
\put(1,5){\line(1,-1){5}}
\put(2,4){\vector(1,-1){0.4}}
\put(4.7,1.3){\vector(1,-1){0.4}}
\put(1,0){\line(1,1){5}}
\put(2,1){\vector(1,1){0.4}}
\put(4.7,3.7){\vector(1,1){0.4}}
\put(6.0,0.1){$\Sigma_1$}
\put(5.3,4.8){$\Sigma_2$}
\put(1.3,4.8){$\Sigma_3$}
\put(1.4,0.1){$\Sigma_4$}
\put(2.8,0.5){$\scriptsize\begin{pmatrix} 1 & - R_1(\zeta) \cdots\\ 0 & 1 \end{pmatrix}$}
\put(4.5,3.1){$\scriptsize\begin{pmatrix} 1 & 0 \\ R_2(\zeta) \cdots & 1 \end{pmatrix}$}
\put(1.9,4.5){$\scriptsize\begin{pmatrix} 1 & - R_3(\zeta) \cdots \\ 0 & 1 \end{pmatrix}$}
\put(0.5,1.8){$\scriptsize\begin{pmatrix} 1 & 0 \\ R_4(\zeta) \cdots  & 1 \end{pmatrix}$}
\end{picture}
\caption{Contours of a cross}
\label{fig:contourcross}
\end{figure}
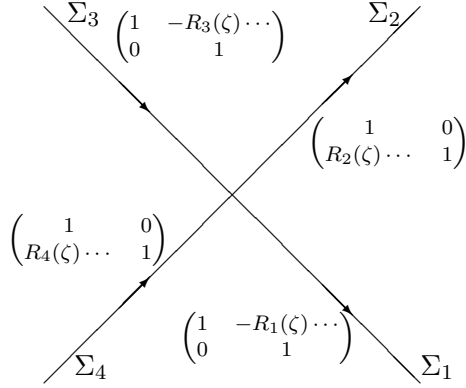

Orient $\Sigma$ such that the real part of $k$ increases
in the positive direction. Denote by $\mathbb{D} = \accol{\zeta\in\D{C}\mid\abs{\zeta}<1}$ the
open unit disc. Throughout this section $\zeta^{\ii\nu}$ will denote
the function $\ee^{\ii\nu\log(\zeta)}$, where the branch cut of the logarithm is chosen along
the negative real axis $(-\infty,0)$.

Introduce the following jump matrices ($v_j$ for $\zeta\in\Sigma_j$, $j=1,\dots,4$)
\begin{equation}
\begin{split}
v_1 &= \begin{pmatrix} 1 & - R_1(\zeta) \zeta^{2\ii\nu} \ee^{- t \Phi(\zeta)} \\ 0 & 1 \end{pmatrix},\quad
v_2 = \begin{pmatrix} 1 & 0 \\ R_2(\zeta) \zeta^{-2\ii\nu} \ee^{t \Phi(\zeta)} & 1 \end{pmatrix},  \\
v_3 &= \begin{pmatrix} 1 & - R_3(\zeta) \zeta^{2\ii\nu} \ee^{- t \Phi(\zeta)} \\ 0 & 1 \end{pmatrix}, \quad
v_4 = \begin{pmatrix} 1 & 0 \\ R_4(\zeta) \zeta^{-2\ii\nu} \ee^{t \Phi(\zeta)}  & 1 \end{pmatrix}
\end{split}
\end{equation}
and consider the RHP given by
\begin{alignat}{2}\label{eq:rhpcross}
M_+(\zeta) &= M_-(\zeta)v_j(\zeta), &\qquad& \zeta\in\Sigma_j,\quad j=1,2,3,4,\\ \notag
M(\zeta) &\to \id, && \zeta\to \infty.
\end{alignat}
The solution is given in the following theorem of P.~Deift and X.~Zhou \cite[Sect.~3-4]{dz} (for a proof of the version stated below see H.~Kr\"uger and G.~Teschl \cite[Theorem A.1]{kt2}).

\begin{theorem}[\cite{dz}]\label{thm:solcross}
Assume there is some $\rho_0>0$ such that $v_j(\zeta)=\id$ for $\abs{\zeta}>\rho_0$ and $j=1,\dots,4$. Moreover,
suppose that within $\abs{\zeta}\leq\rho_0$ the following estimates hold:
\begin{enumerate}[\rm(i)]
\item
The phase satisfies $\Phi(0)=\ii\Phi_0\in\ii\D{R}$, $\Phi'(0) = 0$, $\Phi''(0) = \ii$, and
\begin{align}\label{estPhi}
&\pm \Re\big(\Phi(\zeta)\big) \geq \frac{1}{4}\abs{\zeta}^2,\text{ with }
\begin{cases} 
+ & \text{for } \zeta\in\Sigma_1\cup\Sigma_3,\\ 
- &\text{else},
\end{cases}\\ 
\label{estPhi2}
&\Bigl\lvert\Phi(\zeta) - \Phi(0) - \frac{\ii\zeta^2}{2}\Bigr\rvert\leq C \abs{\zeta}^3.
\end{align}
\item
There is some $r\in\mathbb{D}$ and constants $(\alpha,L)\in(0,1]\times(0,\infty)$
such that $R_1,\dots,R_4$ satisfy H\"older conditions of the form
\begin{equation}
\label{holdcondrj}
\begin{split}
&\abs{R_1(\zeta) - \overline{r}} \leq L \abs{\zeta}^\alpha,\qquad\qquad\quad
\abs{R_2(\zeta) - r} \leq L \abs{\zeta}^\alpha, \\
&\biggl\lvert R_3(\zeta) - \frac{\overline{r}}{1-\abs{r}^2}\biggr\rvert\leq L \abs{\zeta}^\alpha,\qquad
\biggl\lvert R_4(\zeta) - \frac{r}{1-\abs{r}^2}\biggr\rvert\leq L\abs{\zeta}^\alpha.
\end{split}
\end{equation}
\end{enumerate}
Then the solution of the RHP \eqref{eq:rhpcross} satisfies
\begin{equation}
M(k) = \id + \frac{1}{\zeta} \frac{\ii}{t^{1/2}} 
\begin{pmatrix} 
0 & -\beta \\ 
\overline{\beta} & 0 
\end{pmatrix}
+ \ord(t^{- \frac{1 + \alpha}{2}}),
\end{equation}
for $\abs{\zeta}>\rho_0$, where
\begin{equation}
\beta = \sqrt{\nu}\,\ee^{\ii(\pi/4-\arg(r)+\arg(\Gamma(\ii\nu)))}\ee^{-\ii t\Phi_0}t^{-\ii\nu},
\qquad \nu = - \frac{1}{2\pi} \log(1-\abs{r}^2).
\end{equation}
Furthermore, if $R_j(\zeta)$ and $\Phi(\zeta)$ depend on some parameter, the error term is uniform
with respect to this parameter as long as $r$ remains within a compact subset of $\mathbb{D}$
and the constants in the above estimates can be chosen independent of the parameters.
\end{theorem}

\noindent
\textbf{Acknowledgments.}
We are indebted to Ira Egorova for helpful discussions on various topics of this paper and to Adrian Constantin for comments on a previous version of this article. One of us (A.~K.) gratefully acknowledges financial support from the International Erwin Schr\"odinger Institute for Mathematical Physics in the form of a junior fellowship during which parts of this research were done.

\end{document}